\documentclass[11pt,a4paper]{article}
\usepackage[hidelinks]{hyperref}
\usepackage{enumitem}
\setlist{parsep=0pt,topsep=4pt}
\usepackage{authblk}
\usepackage{amsmath,amssymb,amsfonts,amsthm}
\usepackage{pbox}
\usepackage{graphicx}
\graphicspath{{figures/}} 
 \usepackage{xcolor}
\usepackage{thmtools}
\usepackage{thm-restate}
\usepackage{latexsym}
\usepackage{algorithm}
\usepackage[noend]{algpseudocode}
\usepackage{makecell}
\usepackage{xspace}
\usepackage{bm}
\usepackage{microtype}
\usepackage{wrapfig}
\theoremstyle{definition}
\newtheorem{definition}{Definition} 
\newtheorem{theorem}{Theorem}
\newtheorem{lemma}{Lemma}
\newtheorem{corollary}{Corollary}

\let\emptyset\varnothing

\makeatletter
\DeclareRobustCommand*\cal{\@fontswitch\relax\mathcal}
\makeatother

\newcommand{\FSY}{\textsc{Fsynch}\xspace}
\newcommand{\SSY}{\textsc{Ssynch}\xspace}
\newcommand{\RSY}{\textsc{Rsynch}\xspace}
\newcommand{\RR}{\textsc{RR}\xspace}
\newcommand{\CCA}{\ensuremath{1}\xspace}
\newcommand{\PS}{\ensuremath{2}\xspace}
\newcommand{\RA}{\ensuremath{m}\xspace}
\newcommand{\RS}{R}
\newcommand{\RC}{\ensuremath{3}\xspace}
\newcommand{\LU}{\ensuremath{\mathcal{LUMI}}\xspace}
\newcommand{\FS}{\ensuremath{\mathcal{FST\!A}}\xspace}
\newcommand{\FC}{\ensuremath{\mathcal{FCOM}}\xspace}
\newcommand{\OB}{\ensuremath{\mathcal{OBLOT}}\xspace}
\newcommand{\Look}{\ensuremath{\mathit{Look}}\xspace}
\newcommand{\Compute}{\ensuremath{\mathit{Compute}}\xspace}
\newcommand{\Move}{\ensuremath{\mathit{Move}}\xspace}
\newcommand{\LCM}{\ensuremath{\mathit{LCM}}\xspace}
\newcommand{\N}{{\rm I\kern-.22em N}} 
\newcommand{\Z}{{\sf Z\kern-.42em Z}} 
\newcommand{\R}{{\rm I\kern-.22em R}}

\newcommand{\CoP}{{\tt pred}}
\newcommand{\CoS}{{\tt suc}}
\newcommand{\all}{3}
\newcommand{\fst}{4}
\newcommand{\snd}{5} 

\newcommand{\length}{\mathit{length}\xspace}
\newcommand{\status}{\mathit{status}\xspace}
\newcommand{\Status}{\mathit{Status}\xspace}
\newcommand{\light}{\mathit{light}\xspace}
\newcommand{\Light}{\mathit{Light}\xspace}
\newcommand{\executed}{\ensuremath{\mathit{executed}}\xspace}

\newcommand{\checked}{\mathit{checked}\xspace}
\newcommand{\here}{\mathit{here}\xspace}
\newcommand{\lcolor}{\mathit{color}\xspace}
\newcommand{\step}{\mathit{step}\xspace}
\newcommand{\flag}{\mathit{flag}\xspace}
\newcommand{\true}{\mathit{True}\xspace}
\newcommand{\false}{\mathit{False}\xspace}
\newcommand{\cntr}{\mathit{center}\xspace}
\newcommand{\final}{\mathit{final}\xspace}
\newcommand{\other}{\mathit{other}\xspace}
\newcommand{\des}{\mathit{des}\xspace}
\newcommand{\pos}{\mathit{pos}\xspace}
\newcommand{\suc}{\mathit{suc}\xspace}
\newcommand{\charged}{\mathit{charged}\xspace}

\makeatletter

\newlength\tdima
\newcommand\tabfill[1]{%
	\setlength\tdima{\linewidth}%
	\addtolength\tdima{\@totalleftmargin}%
	\addtolength\tdima{-\dimen\@curtab}%
	\parbox[t]{\tdima}{#1\ifhmode\strut\fi}}
\makeatother

\newcounter{Codeline}
\newcommand{\Newcodeline}{\setcounter{Codeline}{1}}
\newcommand{\Cl}{\>{\theCodeline}:\' \addtocounter{Codeline}{1}}
\newcommand{\crm}{\\}

\title{On the Computational Power of Energy-Constrained 
   Mobile Robots: Algorithms and Cross-Model Analysis\thanks{This work was supported in part by 
JSPS KAKENHI No.
20K11685 and 21K11748,  Israel \& Japan Science and Technology Agency (JST) SICORP (Grant\#JPMJSC1806), and by the Natural Sciences and Engineering Research Council of Canada (NSERC) under Discovery Grants A2415 and 203254.}}

\author[1]{Kevin Buchin}
\author[2]{
Paola Flocchini}
\author[1]{Irina Kostitsyna}
\author[1]
{Tom Peters}
\author[3]
{Nicola Santoro}
\author[4]{
Koichi Wada}

\affil[1]{TU Eindhoven, The Netherlands}
\affil[2]{University of Ottawa, Canada}
\affil[3]{Carleton University, Canada}
\affil[4]{Hosei University, Tokyo, Japan}

\date{}
\begin{document}
%
%
%
%
%
\maketitle              
\begin{abstract}
    
 We consider distributed systems of identical autonomous computational entities, called 
 {\em robots}, moving and operating in the plane
 in  synchronous \Look-\Compute-\Move (\LCM)  cycles. 
 The algorithmic capabilities of these systems have been extensively investigated in the literature 
 under four distinct models ($\OB$, $\FS$, $\FC$, $\LU$), each identifying different levels of memory persistence 
 and communication capabilities of the robots. In spite of their differences, they all  always assume that the robots have
 unlimited amounts of energy. 
 
 In this paper we remove this assumption, and start the
 study of the computational capabilities of robots whose energy is limited, albeit renewable.
 More precisely, we consider systems where an activated entity uses all its energy
 to execute an  \LCM  cycle, and the energy  can be restored  through a period of inactivity.

We first study the impact that memory persistence 
and communication capabilities have on the computational power
of such energy-constrained systems of robots; we do so by analyzing 
the computational relationship between the four models under this energy constraint.
We provide a complete characterization of this relationship. 
Among the many results of
this cross-model analysis  we prove that for energy-constrained robots, 
$\FC$  is more powerful  than $\FS$ (that is, it is better to communicate than to remember).
Integral part of the proof is the design and analysis of an algorithm that allows robots in $\FC$ to
execute correctly any protocol for the more powerful $\LU$ model.

We then  study  the difference  in computational power caused by the energy restriction, and provide
a complete characterization of the relationship between  energy-restricted and unrestricted robots
in each of the models. We  prove that, within $\LU$ there is no difference; 
integral part of the proof is the design and analysis of an algorithm that in  $\LU$ allows 
energy-constrained robots  to execute correctly any protocol for robots with unlimited energy.
We then show the (apparently counter-intuitive) result that in all other models,
the energy constraint actually provides the robots with a computational  advantage.

At the basis of our results is the mapping allowing to represent   a system of energy-constrained robots
as  a system of robots with unlimited energy but  subject to a special adversarial activation scheduler.

\end{abstract}


\section{Introduction}

\subsection{Framework: Background and Cross-Model Analysis}

In this paper, we consider distributed systems composed of  identical autonomous computational entities, 
viewed as points and called 
  {\em robots}, moving and operating in the Euclidean plane
 in  synchronous \Look-\Compute-\Move (\LCM)  cycles.

In each synchronous round, a non-empty set of (possibly all) robots is  activated, 
each performs its \LCM cycle simultaneously, and terminates by the end of the round. 
Each cycle is composed of three phases;
in the \Look phase, an entity  obtains a snapshot of the space showing the positions of the other robots; 
in  the \Compute phase, it  executes its algorithm (the same for all robots) using the snapshot as input; 
it then moves towards 
the computed destination in the \Move phase.
Repeating these cycles, the robots are able to collectively perform some tasks and solve some problems.

The selection of which robots are activated in a round is made by an adversarial scheduler.
This general setting is usually called {\em semi-synchronous}  (\SSY); 
the special restricted    setting where every robot is  activated in every round is called  {\em fully-synchronous} 
 (\FSY) .

 These systems have been extensively investigated within distributed computing.
 The research focus has been on understanding
 the nature and the extent of the impact  that  crucial factors,  such as {\em memory persistence} 
 and {\em communication capability},  have  on  the solvability of a problem and thus 
 on the computational power of the system.
 To this end, four models have been identified and investigated: \OB, \FS, \FC, and \LU.

 In the most common (and weakest) model, $\OB$, in addition to the standard assumptions of {anonymity} and {uniformity} 
 (robots have no IDs and run identical algorithms),  the  robots are 
 {\em oblivious} (they have no persistent memory to record information of previous cycles)
 and they are {\em silent} (without explicit means of communication). 
  Computability in this model has been the object of intensive research since its
 introduction in  \cite{SY}; 
 (e.g., see
 \cite{AP,AOSY,BDT,CDN,CFPS,CP,FPSW05,ISKIDWY,
FPSW08,FYOKY,SY,YS,YUKY,
CG,GP,SIW}; as well as the recent
  book and chapters therein \cite{FPS19}).
 Clearly,  the restrictions created by the absence of  persistent memory and the incapacity of
 explicit communication severely limit 
 what the robots can do.

In the stronger  $\LU$ model, formally introduced and defined in \cite{DFPSY},
 robots are provided with some (albeit limited) persistent memory and
communication means.
In this model, each robot is equipped with a constant-sized memory (called {\em light}),
 whose value (called {\em color}) can be set during the \Compute phase. 
 The light  is  visible to all the robots and is persistent in the sense that it does 
 not automatically reset at the end of a cycle. Hence,  these luminous 
 robots  are capable in each cycle of both remembering and communicating a constant number of bits. 
 There is a lot of research work on the design of algorithms and the feasibility of problems for luminous  
 robots  \cite{sssbhagat17,DFPSY,FSVY,HDT,LFCPSV,OWD,OWK,TWK,SABM18,sssgokarna16,V}; 
 for a recent survey, see \cite{DLV19}.
The availability of both persistent memory and communication, however limited, clearly
renders
luminous robots more powerful than  oblivious robots (see e.g., \cite{DFPSY}).


Models \FS and \FC are sub-models of \LU,
introduced in \cite{FSVY} and studied in 
\cite{OWD,TWK}.
They are helpful to understand the individual computational power of persistent memory and communication.
In the first model, $\FS$,  the light of a robot is ``internal'', i.e., visible only to that robot,
 while in the second model, $\FC$,  the light of a robot is visible only to the other robots but not to the robot itself.
Thus in $\FS$  the color merely encodes an internal state;  hence  the robots are
{\em finite-state} and {\em silent}. On the contrary, in $\FC$,  a robot
 can communicate to the other robots through its colored light
  but forgets the content of its transmission by the next cycle; that is,  robots are
 {\em  finite-communication} and {\em oblivious}.  
 
 Summarizing, to understand  the computational power 
of these distributed systems one needs 
to explore and determine the computational power of the robots 
 within each of these   models, 
 as well as (and more importantly) with respect to each other.
  
 This type of cross-model investigation has been taking place but rather limited in scope (e.g., \cite{FSVY}).
 Recently, a  substantial step  has been  taken in \cite{FSW19} where, by integrating existing
 bounds and establishing new results, a comprehensive map  has been drawn 
 of the computational  relationship  between the four models, $\OB$, $\FS$, $\FC$,  $\LU$, 
 (and hence of the computational impact of the presence/absence of persistent memory and/or
 communication capabilities) for  the two fundamental synchronous settings: fully-synchronous and
 semi-synchronous.

\subsection{The Energy Problem}
 
 In the vast existing literature on these systems of autonomous mobile computational entities,
 surprisingly,  no consideration has been made so far on the
 energy required for the robots  to be able to operate.
 In other words, the  existing research  share the same implicit  assumption,
 that the robots have an unlimited amount of energy enabling them to perform
 their activities.

  In this paper we remove this assumption, and start the
 study of the computational capabilities of robots whose energy is limited, albeit renewable.
 More precisely, we consider systems where an activated entity uses all its energy
 to execute an  \LCM  cycle, and that once this happens the robot is not operational 
 and cannot be activated;
  the energy however  can be restored  through a period of inactivity.
 This would be for example the case if the robot's power is provided by a battery 
  rechargeable  by  energy  harvesting (as it is done 
in  conceptually related systems  such as  wireless mobile sensors \cite{ShHJ18}).
 
 The immediate natural questions are: 
 {\em what is the computational power of these energy-constrained robots} ?
and, in particular, 
{\em what is the impact of the crucial factors (memory and communication) in this case} ?

 In this paper we start  investigating these questions. 
%

\subsection{Contributions}

We consider  systems  where
 the energy of a robot  is sufficient for executing exactly one  $\mathit{LCM}$ cycle, and  the depleted energy  is restored  after one round of inactivity. 
We investigate  the computational power of the distributed robotic systems 
 described by the four models  when the robots are subject to such an energy-constraint.

\begin{table}[htbp]
\centering
 \caption{Relationships within \RSY.}
 \label{tab:RS-RS}
        \begin{tabular}{|c|c|c|c|}
\hline
        & \gape{$\mathcal{\FC}^{RS}$}  & $\mathcal{\FS}^{RS}$  & $\mathcal{\OB}^{RS}$   
      \\ \hline
 $\mathcal{\LU}^{RS}$   & \makecell{$\equiv$ \\ (Th.\ref{th:FCRSLURSLUS})} & \makecell{$>$ \\ (Th.\ref{th:FCRSLURSLUS})}  & \makecell{$>$ \\ (Th.\ref{th:FCRSLURSLUS})} \\ 
\hline
  $\mathcal{\FC}^{RS}$  & $-$ & \makecell{$>$ \\ (Th.\ref{th:FCRSLURSLUS})} & \makecell{$>$ \\ (Th.\ref{th:FCRSLURSLUS})}  \\ 
\hline
   $\mathcal{\FS}^{RS}$ & $-$  &$-$ & \makecell{$>$ \\ (Th.\ref{th:FCRSLURSLUS})}  \\
\hline
\end{tabular}
 %
\end{table}

At the basis of our study is the computational correspondence
between   systems  of energy-restricted robots under \SSY\ 
and those same systems  when the robots have unbounded energy 
 (as traditionally considered in the literature) but  whose activation is 
   under the control of  a  special adversarial synchronous scheduler,
 \RSY, never studied before, where  the sets of robots activated in any two 
consecutive rounds are restricted to be disjoint.
This direct correspondence   enables us to reduce the 
cross-model investigation 
of these energy-constrained robots
to  the  cross-model investigation of energy-unbounded robots
under the \RSY\ scheduler.
Furthermore, it allows us to  determine  the change (if any) in
computational power due to the energy restrictions, by
determining the relationship between  \RSY\ and the 
general unrestricted \SSY\ scheduler.

Let 
 $M^{S}$  and $M^{RS}$
 denote the systems  of unlimited-energy robots 
 defined by model ${M}\in \{\LU, \FC,\FS,\OB \}$,
 under \SSY\ and \RSY, respectively (the latter being
 equivalent to the systems  of constrained energy robots
 under \SSY).
 \color{black}
 


We first study the impact that memory persistence 
and communication capabilities have on the computational power
of  energy-constrained systems of robots; we do so by analyzing 
the computational relationship between the four models under this energy constraint.
We provide a complete characterization of this relationship proving that (see  Table~\ref{tab:RS-RS}):
 $$\mathcal{\LU}^{RS} \equiv \FC^{RS} > \FS^{RS} > \OB^{RS}$$
 where (as formally defined in Sect. \ref{sec:Relation}) 
 $\cal{X} > \cal{Y}$ denotes that $\cal{X}$ is strictly more powerful than $\cal{Y}$,
  $\cal{X} \equiv \cal{Y}$ denotes that $\cal{X}$ and $\cal{Y}$ are
  computationally equivalent, and $\cal{X} \bot \cal{Y}$,
  denotes that $\cal{X}$ and $\cal{Y}$ are computationally incomparable.
  
Integral part of the proof  that $\FC^{RS}$  is more powerful  than $\FS^{RS}$ 
(that is, it is better to communicate than to remember), is the design and analysis of an algorithm 
that allows robots in $\FC^{RS}$ to
execute correctly any protocol for the more powerful $\LU^{RS}$.

We then  study  what  variation   in computational power is created by the 
presence of the energy restriction, by comparing the computational difference
between energy- restricted and unrestricted robots in each of the four models
(i.e., between  $M^{RS}$ and $M^{S}$ for each  ${M}\in \{\LU, \FC,\FS,\OB \}$).
We provide a complete characterization (see Table \ref{tab:RS-SS}). 
In particular, we  prove that 
within $\LU$, the strongest model, there is no difference;
 i.e., $\LU^{RS}\equiv \LU^{S}$.
Integral part of the proof is the design and analysis of an algorithm that 
allows 
energy-constrained robots  to execute in $\LU^{RS}$ correctly any protocol 
for robots with unlimited energy $\LU^{S}$.
 In all other models, we  prove   that
the energy constraint actually provides the robots with a 
definite computational  advantage;  this apparently counter-intuitive result
is due to the fact that  the energy restriction  reduces the adversarial power 
of the activation scheduler.
Let us stress that the established characterization  covers all  the  cross-model and
cross-scheduler relationships.

\begin{table}[htbp]
\centering
\caption{Relationship between \RSY\ and \SSY.}\label{tab:RS-SS}
\begin{tabular}{|c|c|c|c|c|}
\hline
      & \gape{$\mathcal{\LU}^{S}$}  & $\mathcal{\FC}^{S}$  & $\mathcal{\FS}^{S}$  & $\mathcal{\OB}^{S}$ \\
\hline
 \gape{\makecell{$\mathcal{\LU}^{RS}$ \\ $\equiv \mathcal{\FC}^{RS}$}}  & \makecell{$\equiv$ \\ (Th.\ref{th:FCRSLURSLUS})} & \makecell{$>$ \\ (Th.\ref{th:RSvsS1})} & \makecell{$>$ \\ (Th.\ref{th:FCRSLURSLUS}, Th.\ref{th:RSvsS1})} & \makecell{$>$ \\ (Th.\ref{th:FCRSLURSLUS}, Th.\ref{th:RSvsS1})}\\
\hline
%
%
%
    $\mathcal{\FS}^{RS}$ & \makecell{$<$ \\ (Th.\ref{th:FCRSLURSLUS})} & \makecell{$\bot$ \\ (Th.\ref{th:RSvsS2})} & \makecell{$>$ \\ (Th.\ref{th:RSvsS1})} & \makecell{$>$ \\ (Th.\ref{th:FCRSLURSLUS}, Th.\ref{th:RSvsS1})}\\
\hline
    $\mathcal{\OB}^{RS}$ & \makecell{$<$ \\ (Th.\ref{th:FCRSLURSLUS})} & \makecell{$\bot$ \\ (Th.\ref{th:RSvsS1})} & \makecell{$\bot$ \\ (Th.\ref{th:RSvsS1})} & \makecell{$>$ \\ (Th.\ref{th:RSvsS1})}\\
\hline
\end{tabular}
\end{table}

\begin{table}[htbp]
\begin{center}
\caption{Relationship between \FSY\ and \RSY.}\label{tab:FS-RS}
\begin{tabular}{|c|c|c|c|}
\hline
      & \gape{$\mathcal{\LU}^{RS} \equiv \mathcal{\FC}^{RS}$}   & $\mathcal{\FS}^{RS}$  & $\mathcal{\OB}^{RS}$  \\ \hline
      
 \makecell{$\mathcal{\LU}^{F}$ \\ $\equiv \mathcal{\FC}^{F}$} & \makecell{$>$ \\ (Th.\ref{th:FvsRS})} & \makecell{$>$ \\ (Th.\ref{th:FvsRS}, Th.\ref{th:FCRSLURSLUS})}
 & \makecell{$>$ \\ (Th.\ref{th:FvsRS}, Th.\ref{th:FCRSLURSLUS})}\\
\hline
  
  
   $\mathcal{\FS}^{F}$ & \makecell{$\bot$ \\ (Th.\ref{th:FSFneqLURSFCRS})} & \makecell{$>$ \\ (Th.\ref{th:FvsRS})} & \makecell{$>$ \\ (Th.\ref{th:FvsRS}, Th.\ref{th:FCRSLURSLUS})} \\
\hline
    $\mathcal{\OB}^{F}$ & \makecell{$\bot$ \\ (Th.\ref{th:orthFSOBwithRS})} & \makecell{$\bot$ \\ (Th.\ref{th:orthFSOBwithRS})} & \makecell{$>$ \\ (Th.\ref{th:FvsRS})}\\
\hline
\end{tabular}
\centering
\end{center}
\end{table}

Finally, we  complete the study of  systems of energy-restricted robots 
by analyzing the relationship between their computational   power and that
of  robots with unlimited energy  under the most benign
synchronous activation scheduler \FSY\ (i.e., fully synchronous).
In this case, perhaps not surprising,  
we prove that, in each model, energy-constrained robots are strictly
less powerful than fully-synchronous ones with unbounded energy
(see Table \ref{tab:FS-RS}). Also in this study, 
the characterization  covers all  the  cross-model and
cross-scheduler relationships.

An interesting consequence of our investigation is that it provides a novel insight
in the landscape of synchronous activation environments. 
In fact, it identifies  the presence of a 
novel  synchronous environment which lies in between the
semi-synchronous and the fully synchronous ones.
It is worth noting that  
the computational  cross-model ``map''  of  the new scheduler \RSY\ 
(which includes {\em Round Robin} as a particular case) has
the same structure as that of \FSY.
A graphical summary of all results and of these observations  is shown in Figure \ref{fig:summary}.


\begin{figure}[tbp]
\centering
\includegraphics{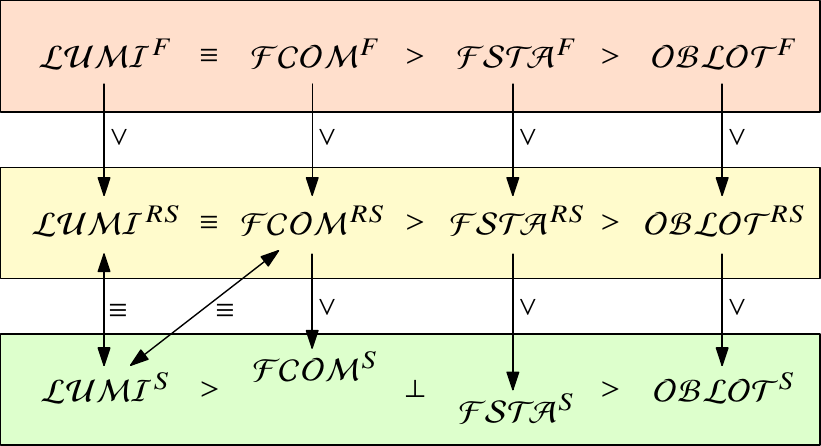}
    \caption{Summary of relationships.}
    \label{fig:summary}
\end{figure}


\section{Models and Preliminaries}\label{sec:model}

\subsection{The Basics}

 Each  robot operates in \Look-\Compute-\Move (\LCM) cycles: it observes
its surroundings, it computes a destination within the space  based on what it sees, and it moves toward the destination.
 It can observe its surroundings and move within the space based on what it sees. The same space may be populated by several mobile robots, each with its local coordinate system, and static objects.

The system  consists of 
a set  $R = \{ r_0,\cdots,r_{n-1}\}$ 
of  computational entities, called robots,  
 modeled as  geometric points, that live in  $\mathbb R^2$,
 where they can move freely and continuously.
 The robots are autonomous without a central control.
They  are indistinguishable by their appearance,
do not have internal identifiers, and  execute 
the same algorithm.
 
 Let $x_i(t)$ denote the location of robot $r_i$ at time $t$ 
in a global coordinate system (unknown to the robots), and let
$X(t)= \{x_i(t) : 0 \leq i \leq {n-1}\} =  \{x_0(t),x_1(t), \ldots, x_{m-1}(t)\}$; 
observe that $|X(t)| = m \leq n$ since several robots might be
at the same location at  time $t$.

Each robot has its own local coordinate system and it  perceives itself at its origin.
A robot is equipped with  devices that allow it to 
observe  
the positions of the other robots in its local coordinate system.

The robots operate in \Look-\Compute-\Move (\LCM) 
cycles.  When activated,
a robot  executes a 
cycle by performing the following three operations:
\begin{enumerate}
\item {\em Look:} The robot  obtains a snapshot of the positions occupied by robots 
expressed with respect to its own coordinate
system; this  operation is assumed to be instantaneous.
\item {\em Compute:} The robot executes the algorithm using the snapshot as input;
the result of the computation is a destination point.
\item {\em Move:} The robot moves towards  the computed destination.
If the destination is the current location, the robot stays still.
\end{enumerate}

The system is {\em synchronous}; that is, time is divided into discrete
intervals, called {\em rounds}. In each round a robot is either active or inactive.
The robots  active in a round
perform their \LCM cycle in perfect synchronization; 
if not active, the robot is idle in that round. 
All robots are initially idle. 
In the following,  we use round and time interchangeably.

Each robot has a bounded amount of {\em energy}, which is totally consumed 
whenever it performs a cycle; its energy however is restored after being 
idle for a round. A robot  with depleted energy cannot be active.

Movements are said to be {\em rigid} if the robots always reach their destination.
They are said to be {\em non-rigid}  if they may be unpredictably stopped by an adversary 
whose only limitation is the existence of $\delta>0$, unknown to the robots, such that
  if the destination is at distance at most $\delta$ the robot will  
reach it, else it will move at least $\delta$ towards the
  destination.

There might not be consistency between the local coordinate systems and their  unit of distance. 
The absence of any a-priori assumption on consistency of the local coordinate systems is called {\em disorientation}.
The type of disorientation can range from  {\em fixed}, where each local coordinate system remains the same through all the rounds,
to {\em variable} where   the direction, the orientation, and the unit of distance 
of a robot may vary between successive rounds.    In this paper we consider only fixed disorientation.


The robots are said
to have {\em chirality} if they share the same circular orientation of the plane (i.e., they 
agree on ``clockwise'' direction). Notice that, 
in presence of  chirality,  at any time~$t$, there would exist
a unique  circular ordering of the locations $X(t)$ occupied by the robots at that time;
let  {\tt suc} and {\tt pred}  be the functions  denoting the ordering and, without loss of generality,
let
{\tt suc}$(x_{i}(t))=x_{i+1\bmod m}(t)$ and {\tt pred}$(x_i)(t)=x_{i-1\bmod m}(t)$ for $i \in \{0,1, \ldots,m-1\}$.

\subsection{The Computational Models}

In the most common model, $\OB$, the robots are {\em silent}: they have no explicit means of communication; furthermore they are {\em oblivious}: at the start of a cycle, a robot has no
memory of observations and computations performed in previous cycles.

In the other common model, \LU, 
each robot $r$ is  equipped with a persistent  visible
state variable $\Light[r]$, called {\em light}, whose values are taken from a finite set $C$ of states called {\em colors} (including the color that represents the initial state when the light is off). 
The colors of the lights can be set in each cycle by $r$ at the end of its \Compute operation. 
A light is {\em persistent} from one computational cycle to the next: the color is not automatically reset at the end of a cycle;  the robot is otherwise oblivious, forgetting all other information from previous cycles.
In \LU, the \Look operation produces a colored snapshot; i.e., it returns the set of pairs 
 $(\mathit{position},\lcolor)$ 
of the other robots\footnote{If  (strong) multiplicity detection is assumed, the snapshot is a multi-set.}.
Note that if $|C|=1$, then the light is not used; thus, this case corresponds to the $\OB$ model. 

It is sometimes convenient to describe a robot $r$ as having $k\geq 1$ lights, denoted
 $r.\light_1,\ldots,\\r.\light_k$, where the values of $r.\light_i$ are from a finite set of colors
 $C_i$, 
 and to consider $\Light[r]$ as a $k$-tuple of variables; clearly, this corresponds to $r$ having a
 single light that uses $\Pi^{k}_{i=1}|C_i|$ colors.

The lights provide simultaneously   persistent memory and 
direct means of communication, although both limited to a constant number of bits per cycle.
 Two sub-models of $\LU$ have been defined and investigated, each offering only one
 of these two capabilities.
 
 In the first model, $\FS$,  a robot can only see the color of its own light; that is, the light is
 an {\em internal} one and its color merely encodes an internal state.
 Hence  the robots are {\em silent}, as in $\OB$; but 
  are {\em finite-state}. Observe that a snapshot in $\FS$ is the same as in $\OB$.

In the second model, $\FC$,  the lights are {\em external}: a robot
 can communicate to the other robots through its colored light
  but forgets the color of its own light by the next cycle; that is,  robots are
 {\em  finite-communication} but  {\em oblivious}.  
 A snapshot in $\FC$ is like in \LU except that, for the position $x$ where the robot
 $r$ performing the $Look$ is located,  $Light[r]$ is omitted from the set 
 of  colors present at $x$.

 In all the above models,
 a {\em configuration} $\textit{C}(t)$ at time $t$ is the multi-set of  the $n$ pairs 
 of the 
 $(x_i(t), c_i(t))$, where $c_i(t)$ is the color of robot $r_i$ at time $t$. 

\subsection{Activation Schedulers and Energy Restriction}

In all computational models just described, 
in each synchronous round, some robots become active 
and they  execute their \LCM cycle in complete synchrony.
The choice of which non-empty subset of the
robots is activated in a specific round is
considered to be under the control of
an adversarial {\em activation scheduler}
 constrained  to be fair; that is, every robot will become active
infinitely often.

Given a synchronous scheduler $\cal S$ and a  set of robots 
$R$, an {\em activation  sequence}  of $R$ by (or, under) $\cal S$
  is  an infinite sequence  $E= \langle e_1, e_2, \ldots, e_i, \ldots \rangle$,
where   $e_i\subseteq R$ denotes the set of robots
activated in round $i$,  satisfying   the {\em fairness constraint}:
$$[(\forall r\in R\; \exists e_i : r\in e_i)  \textbf{ and}\  (\forall i \geq 1,  r\in e_i  \Rightarrow  \exists j>i : r\in e_j)].$$
Let ${\cal E}({\cal S},R)$ denote the set of all activation sequences  of $R$ by $\cal S$.

In the standard  synchronous scheduler (\SSY), first studied in \cite{SY}
 and often called  {\em semi-synchronous},  each sequence  
 $E= \langle e_1, e_2, \ldots, e_i, \ldots \rangle\in {\cal E}(\SSY,R)$
 satisfies the {\em basic condition}
\begin{equation} \forall i\geq 1,  \emptyset \neq e_i\subseteq R.\end{equation}

 The special  {\em fully-synchronous} (\FSY) setting,
 where every robot is  activated in every round, 
 corresponds to  further
 restricting   the activation sequences 
by imposing  $\forall i\geq 1, e_i=R$. 
Notice that, in this setting,
  the activation scheduler has no adversarial power. 
 Another special setting is   defined by the 
well-known {\em round-robin}  (\RR)
scheduler (see, e.g., \cite{DePP20,TWK}), whose generalized definition corresponds to
 adding the restriction:  $[\exists p>1 : (\cup_{1\leq i\leq p} e_i = R)\ \textbf{ and}\
(\forall 1\leq i\neq j\leq p, [e_i \cap e_{j}=\emptyset])\
   \textbf{ and}\ (\forall  i\geq 1,  [e_i=e_{i+p}])]$.

We study systems 
 of  {\em energy-constrained robots} under the standard synchronous
 activation scheduler.
 More precisely, 
 we study  systems where  a robot (i) has just enough energy to execute a cycle,
(ii) it cannot be activated in a round unless it has full energy, and
(iii) its depleted energy is  regenerated after one round.

These three conditions  clearly have an impact  on  
 the possible activation sequences of the robots. 
   In particular,
since a robot with depleted energy cannot be activated,
 the basic condition on  $e_i$   becomes 
\begin{equation} \forall i\geq 1, e_i\subseteq R^*[i] \textbf{ and } R^*[i]\not=\emptyset \Rightarrow e_i\not=\emptyset\end{equation}
where $R^*[i] \subseteq R$ denotes the set of robots with full energy
in round $i$.  Furthermore, since it takes   a round to
regenerate depleted energy, $e_i$ must also satisfy
\begin{equation}\forall i\geq 1, (e_i \cap e_{i+1} = \emptyset)\end{equation}
 
 Notice that, since $e_i\subseteq R^*[i]$, 
 it is possible that ${e}_ {i}= R$ when  $R^*[i] =R$. Should this be the case,
 since a robot  has just enough energy to execute a cycle, then
 $R^*[i+1]=\emptyset$ and, since a robot with depleted energy cannot be activated,
 ${e}_ {i+1}=\emptyset$.
Further observe  that, due to the conditions imposed by
the energy limitations, 
    if $\emptyset\neq e_i\neq R$ then 
\[
\forall j\geq i,  [\emptyset\neq e_j\neq R\; \textbf{ and}\
(e_j \cap e_{j+1} = \emptyset)]\,.
\]
that is, if fewer than $|R|$  full energy robots are activated in any round $i$, 
then $R^*[j] \neq R$ for all $j\geq i$.

In other words,  
 the activation sequences of the
  energy-constrained robots 
are infinite sequences where the prefix is a  (possibly empty)
  alternating sequence of $R$ and $\emptyset$, and, if the prefix is finite, the rest are non-empty sets
  satisfying the constraint $(e_i \cap e_{i+1} = \emptyset)$.

  Notice that this set of sequences, denoted by ${\cal E}(\SSY_\mathit{res},R)$,
  is not a proper subset of ${\cal E}(\SSY,R)$ since some sequences might have 
  empty sets in their prefix.

  Consider now  the synchronous  scheduler, we shall call \RSY,  obtained from \SSY\ by
adding the following {\em restricted-repetition
  condition} to its activation sequences:
\[
\begin{aligned}
	& [\forall i\geq 1, e_i=R]\  \textbf{ or } \\
	& [\exists p\geq 0 : ( [\forall  i\leq p,  (e_i=R)]\ \textbf{ and }\\
	& \phantom{[\exists p\geq 0 : ( } [\forall  i > p,( \emptyset\neq e_i\neq R\ \textbf{ and }\  e_i \cap e_{i+1} = \emptyset)])]\,,
\end{aligned}
\]
 \noindent that is,  ${\cal E}(\RSY,R)$ is composed of sequences where
 the prefix is a  (possibly empty) sequence of $R$ and, if the prefix is finite, the rest are non-empty sets
  satisfying the constraint $(e_i \cap e_{i+1} = \emptyset)$.

There is an obvious bijection $\phi$\ between ${\cal E}(\SSY_{res},R)$\ and 
 ${\cal E}(\RSY,R)$, where $\phi(E)$ corresponds to
  removing all empty sets from $E\in{\cal E}(\SSY_{res},R)$; to this 
corresponds the obvious identity between
 the computation performed by $R$ under  
$E$ and that performed under
$\phi(E)$.
Informally, under $\SSY_{res}$, if all robots are activated in the same round $i$, 
they will be all idle in  round $i+1$, and they will all be with full energy
in  round $i+2$.  Since no activity takes place in  round $i+1$,
the selected set $e_{i+2}$ will perform exactly the same computation
 as if they had been activated in round $i+1$.

In other words, the computation by  energy-constrained $R$ under 
the standard synchronous scheduler \SSY.
is the same as it would be if the robots in $R$ were energy-unbounded
but the activation was controlled by scheduler \RSY.

This   restricted-repetition   setting has never been studied before;
 observe that it includes both fully synchronous \FSY\ and round robin \RR\
 as special cases.

\subsection{Computational Relationships}
\label{sec:Relation}

Let ${\cal M} = \{\LU, \FC,\FS,\OB \}$  be the set of models under investigation, and 
${\cal S}= \{\FSY, \RSY, \SSY\}$ be the set of activation schedulers under consideration.

We denote by $\mathcal{R}$ the set of all teams of robots satisfying the core assumptions (i.e., they are identical, autonomous, and operate in \LCM cycles), and $R \in \mathcal{R}$ a team of robots having identical capabilities (e.g., common coordinate system, persistent storage, internal identity, rigid movements etc.). By $\mathcal{R}_n \subset \mathcal{R}$ we denote the set of all teams of size $n$.

Given a model  $M \in {\cal M}$, a scheduler $S\in  {\cal S}$, and a team of robots $R \in \mathcal{R}$, 
let $Task(M,S;R)$ denote the set of problems solvable by $R$ in  $M$ 
under adversarial scheduler $S$.

Let $M_1, M_2\in{\cal M}$ and $S_1, S_2\in{\cal S}$.
\begin{itemize} 
\item We say that  model $M_1$  under scheduler  $S_1$
is {\em computationally not less powerful than} 
model $M_2$  under  $S_2$, denoted by
$M_{1}^{S_1} \geq M_{2}^{S_2}$
 if $\forall R \in \mathcal{R}$
  we have $Task(M_1,S_1;R) \supseteq Task(M_2,S_2;R)$.

\item We say that 
 $M_1$  under  $S_1$
is {\em computationally more powerful than} 
$M_2$  under  $S_2$, 
denoted by
$M_{1}^{S_1} >  M_{2}^{S_2}$,
if 
$M_{1}^{S_1} \geq M_{2}^{S_2}$ 
and
$\exists R \in \mathcal{R}$ such that 
$Task(M_1,S_1;R) \setminus Task(M_2,S_2;R)  \neq \emptyset$.

\item We say that  $M_1$  under  $S_1$
and  $M_2$  under  $S_2$ are {\em computationally equivalent}, denoted by  
$M_{1}^{S_1} \equiv  M_{2}^{S_2}$,
if $M_{1}^{S_1} \geq M_{2}^{S_2}$ and $M_{2}^{S_2} \geq M_{1}^{S_1}$.

\item Finally, we say that 
 $M_1$  under  $S_1$
and  $M_2$  under  $S_2$ are {\em computationally orthogonal} (or {\em incomparable}), denoted by  
$M_{1}^{S_1} \bot  M_{2}^{S_2}$,  if $\exists R_1, R_2 \in \mathcal{R}$
 such that $Task(M_1,S_1;R_1) \setminus Task(M_2,S_2;R_1) \neq \emptyset$ and 
 $Task(M_2,S_2;R_2) \setminus Task(M_1,S_1;R_2) \neq \emptyset$.

 \end{itemize}
 
 For simplicity of notation, for a model $M\in{\cal M}$, let $M^{F}$, $M^{RS}$, and $M^{S}$ denote $M^{\FSY}$, 
 $M^{\RSY}$, and $M^{\SSY}$, respectively;
 and let $Task(M,\FSY;R)$,
 $Task(M,\RSY;R)$, and $Task(M,\SSY;R)$
 be denoted by let  $M^F(R), M^{RS}(R),$ and $M^S(R)$, respectively.
  
 
Trivially,
 for all $M\in{\cal M}$,
\[
M^{F} \geq M^{RS} \geq M^{S}\,,
\]
also, for all $P\in{\cal S}$,
\[
\begin{aligned}
\LU^{P}&\geq\FS^{P}\geq\OB^{P}\,, \text{ and }\\
\LU^{P}&\geq\FC^{P}\geq\OB^{P}\,.
\end{aligned}
\]

 The following theorems, established in \cite{FSW19}, summarize  the 
 computational relationship between \SSY\ and \FSY.
   
\begin{theorem}[\cite{FSW19}]\label{th:InFandS}~
\begin{enumerate}
\item\label{LUF=FCF>FSF>OBF} $\LU^{F} \equiv \FC^{F} > \FS^{F} > \OB^{F}$,
\item\label{FCSneqFSS} $\FC^{S} \bot$ $\FS^{S}$,
\item\label{LUS>FSS>OBSandLUS>FCS>OBS} $\LU^{S} > \FS^{S} >  \OB^{S}$  and\\  $\LU^{S} > \FC^{S} > \OB^{S}$.
\end{enumerate}
\end{theorem}

\begin{theorem}[\cite{FSW19}]\label{th:FSYNCvsSSYNC}~
\begin{enumerate}
    \item\label{LUF=FCF>LUS>FCS} $ \LU^{F} \equiv$ $\FC^{F} > \LU^{S}  > \FC^{S} $,
    \item\label{FSFneqLUS} $\FS^{F} \bot$ $\LU^S$,
    \item\label{FSF>FSS} $\FS^{F} > \FS^{S}$,
    \item\label{OBF>OBF} $\OB^{F} > \OB^{S}$,
    \item\label{OBFneqLUSFCSFSS}  $\OB^{F} \bot$ $\LU^{S}$, $\OB^{F} \bot$ $\FC^{S}$, and\\  $\OB^{F} \bot$ $\FS^{S}$.
\end{enumerate}
\end{theorem}

These theorems hold assuming rigidity and chirality and these results are extended and they hold without rigidity and chirality~ \cite{apdcm}. 




  \section{Computational Relationship between \RSY\ and \SSY}

In this section we analyze the impact that the presence of the energy constraints
has on the computational capability of the robots. We do so by studying
the computational relationship between \RSY\ and \SSY in each of the four models.

We   show that,  
 $$\mathcal{X}^{RS} > \mathcal{X}^{S} (\forall \mathcal{X} \in \{\FC, \FS, \OB \}),$$
and  $\LU^{RS} \equiv \LU^{S}$.
We also show that 
  the power of \RSY\
makes $\FC^{RS}$ as powerful as  $\LU^{S}$, and even as $\LU^{RS}$ in Section~\ref{LUS=FCRS}.

\subsection{Power of \RSY\ in \texorpdfstring{$\FC,\FS$ and $\OB$}{FCOM,FSTA and OBLOT}}

In this section, we study the relationships among the various models in \RSY\ and we show that  for ${\cal X}\in\{\FC,\FS,\OB\}$ we have
${\cal X}^{RS}>{\cal X}^{S}$.


We start by showing  that Rendezvous problem (RDV)~\cite{FPS}, where two robots $a$ and $b$ must gather in the same location not known in advance, cannot be solved  in \RSY. 

\begin{lemma}\label{lem:RDV}
    $\exists R \in \mathcal{R}_2$, $\text{RDV} \not \in \OB^{RS}(R)$.
    This result holds even in presence of   chirality and rigidity of movement.
\end{lemma}
\begin{proof} 
%
Consider two robots $a$ and $b$ with the same  chirality. The two robots have exactly the same view of the universe.
Assume by contradiction that $A$ is a solution
protocol. Consider an execution {\cal E}  of $A$ where, in the last round, only one robot, say $a$, is activated and moves (achieving rendezvous for the first time) while the other robot does not move in that
round (such an execution can be shown to exist). Consider now the execution up to (and excluding)
the last round; at this point, proceed by activating not just one (as in  {\cal E}), but both robots. When
this happens, $a$ will perform the same move as in {\cal E}, whose destination is the observed position
of $b$. Since the view of $b$ is specular, once activated, $a$  will choose the observed position of $b$  as its
destination. The result will be just a switch of  positions of the robots. Since the robots are oblivious,
in the same conditions they will repeat the same actions; this means that, if they are both activated
in every turn from now on, they will continue to switch without ever gathering.
%
\end{proof}

The following problem was introduced to show that $\OB^{F}$ is computationally more powerful than $\OB^S$, and that models $\OB^{F}$ and $\FC^S$ (or $\FS^S$) are incomparable~\cite{FSW19}.
The problem can also take a role to show $\OB^{RS} > \OB^S$ and orthogonality of $\OB^{RS}$ and $\FC^S$ (or $\FS^S$).

\begin{figure} [tbh]
\centering
\includegraphics{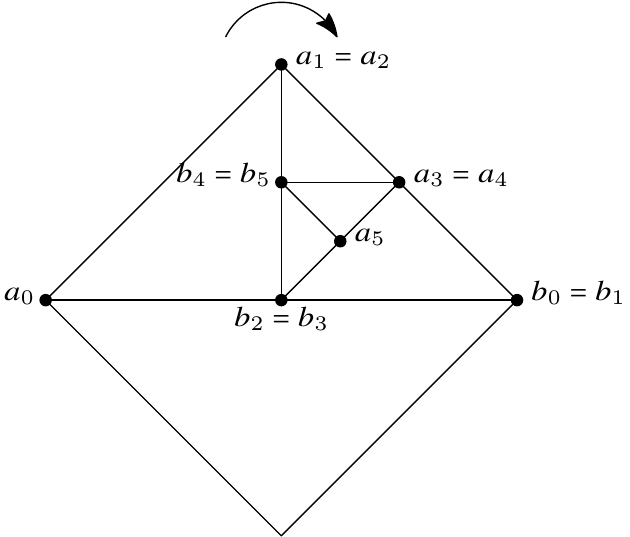}
    \caption{Shrinking Rotation (SRO). Configurations $C_0,C_1,C_2,C_3,\dots$ where $C_i=(a_i,b_i)$.}
    \label{fig:RNER}
\end{figure}

\begin{definition}{\bf SHRINKING ROTATION (SRO)}~\cite{FSW19}:
Two robots $a$ and $b$ are  initially  placed in arbitrary distinct points (forming the initial configuration $C_0$).
The two robots uniquely identify a square (initially $Q_0$) whose diagonal is given by the segment  between them\footnote{By  square, we means the entire space delimited by the four sides.}. 
Let $a_0$ and $b_0$ indicate the initial positions of the robots, $d_0$  the segment between them, and $\length(d_0)$ its length.
Let $a_i$ and $b_i$ be the  positions of $a$ and $b$ in   configuration $C_i$ ($i \geq 0$). 
The problem consists of moving from  configuration $C_i$  to $C_{i+1}$ in such a way that 
either Condition~1 or Condition~2 is satisfied, and Condition~3 is satisfied in either case:

\begin{itemize}
\item Condition~1: $d_{i+1}$ is a $90^\circ$ clockwise rotation of $d_i$
and thus  $\length(d_{i+1})= \length(d_i)$,
\item Condition~2: $d_{i+1}$ is a ``shrunk'' $45^\circ$ clockwise rotation of $d_i$ 
such that $\length(d_{i+1}) = \frac{\length(d_i)}{\sqrt{2}}$, 
\item Condition~3: $a_{i+1}$ and $b_{i+1}$ must be contained in the square $Q_{i-1}$\footnote{We define $Q_{-1}$ to be the whole plane.}.
\end{itemize}
\end{definition}

\Newcodeline
\begin{algorithm}[tp]
\caption{{\tt Alg-SRO} for robot $r$.}
\label{alg:RNER}
\small 
\begin{tabbing}
111 \= 11 \= 11 \= 11 \= 11 \= 11 \= 11 \= \kill
{\em Assumptions}: $\OB$, \RSY\crm
\crm
{\em State Look}\crm
\> $r.\pos$, $\other.\pos$: positions of robot $r$ and the other robot;\crm
\crm
{\em State Compute}\crm

\Cl \tabfill{$r.\des \leftarrow$ the point of clockwise rotating by $90^\circ$ around the midpoint of $r.\pos$ and $\other.\pos$}\crm

\crm

{\em State Move}\crm
\> Move to $r.\des$
\end{tabbing}
\end{algorithm}

%

\begin{lemma}\label{lem:SRO-OBRS}
$\forall \mathit{R} \in \mathcal{R}_2, \text{SRO}  \in \OB^{RS}(R)$, 
assuming common chirality and rigid movement.
\end{lemma}
\begin{proof}
The proof is  by construction: Algorithm~\ref{alg:RNER}   prescribes a robot to rotate clockwise by $90^\circ$ with respect to the midpoint between the two robots.
Note that any schedule in \RSY\  allows  consecutive simultaneous activation of the two robots, but as soon as one of them is activated alone during a round, the subsequent  rounds necessarily consist of  an alternation of activations of each. So, there are only two possible types of executions under \RSY\ with two robots: $(1)$ a perpetual activation of  both robots in each round, or $(2)$ the activation of both for a   finite number of rounds, followed by a perpetual  alternation of activations of each.
In Case (1), the problem is clearly solved by  Algorithm~\ref{alg:RNER}   because the robots keep rotating of $90^\circ$ clockwise around their mid-point,    fulfilling Conditions~1 and 3. 
In Case (2), once the execution becomes an alternation of activations,  the robots  move alternately,  achieving, at each movement, a $45^\circ$ rotation of the segment between them  decreasing its size precisely as required to fulfill  Condition~2 and 3. In fact, in  doing so, they move within a shrinking area where each new square is included in the previous (in a kind of fractal way, see Figure~\ref{fig:RNER}). 
%
Then SRO can be solved with $\OB$ in \RSY.
\end{proof}
 
\begin{lemma}[\cite{FSW19}]\label{lem:SRO-FCSFSSimpo}
 $\exists \mathit{R} \in \mathcal{R}_2, \text{SRO} \not \in \FC^{S}(R) \cup \FS^{S}(R)$.
 This result holds even in presence of   chirality and rigidity of movement.
\end{lemma}

We have seen that SRO can be solved in $\OB^{RS}$ but cannot be solved in $\FC^S$ and $\FS^S$.
On the other hand, RDV can be solved in $\FC^{S}$ and $\FS^S$~\cite{FSVY}  but cannot be solved in $\OB^{RS}$ by Lemma \ref{lem:RDV}.
Thus $\OB^{RS}$ and $\FC^{S}$ ($\FS^S$) are orthogonal.
Since $\FC^{RS}$ and $\FS^{RS}$ includes $\OB^{RS}$, we can show that ${\cal X}^{RS}>{\cal X}^{S}$ for ${\cal X}\in\{\FC,\allowbreak\FS,\OB\}$
by Lemmas~\ref{lem:SRO-OBRS} and~\ref{lem:SRO-FCSFSSimpo}.

\begin{theorem}\label{th:RSvsS1}~
 \begin{enumerate}
    \item\label{OBRSneqFCS} $\OB^{RS} \bot$ $\FC^{S}$,
   \item\label{OBRSneqFSS} $\OB^{RS} \bot$ $\FS^{S}$,
   \item\label{FCRS>FCS} $\FC^{RS} > \FC^{S}$, 
   \item\label{FSRS>FSS} $\FS^{RS} > \FS^{S}$,  
   \item\label{OBRS>OBS} $\OB^{RS} > \OB^{S}$.    


\end{enumerate}
\end{theorem}

It remains to determine the relationship between $\FC^{RS}$ and $\LU^{S}$, and between $\FS^{RS}$ and $\LU^{S}$ or $\FC^{S}$.
The equality $\FC^{RS} \equiv$ $\LU^{S}$ follows from the equalities $\LU^{S} \equiv$ $\LU^{RS}$ and $\LU^{S} \equiv$ $\FC^{S}$, which we show constructively in Sections~\ref{LUS=LURS} and~\ref{LUS=FCRS}.
We also show the dominance of $\LU^{S}$ over $\FS^{RS}$  and  the orthogonality of $\FS^{RS}$ with $\FC^{S}$ in the following.

In order to show these results, 
we use the following problem called Cyclic circles (CYC).

\begin{figure}[H]
\centering
\includegraphics{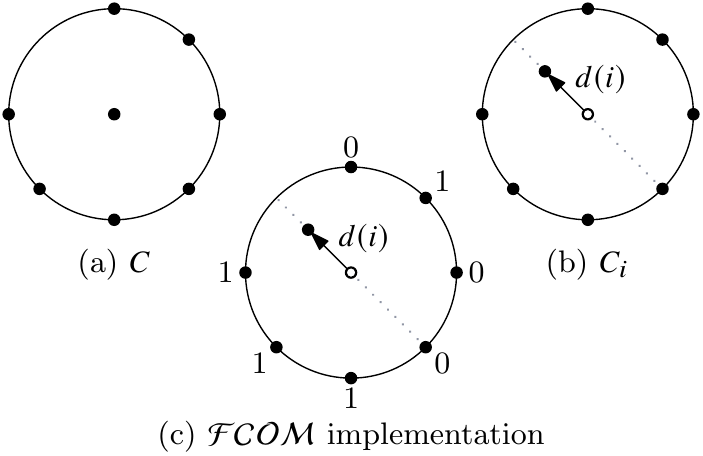}
    \caption {The configurations of problem CYC.}
    \label{Fig:CC}
  \end{figure}
 \begin{definition}
\noindent  {\bf CYCLIC CIRCLES (CYC)}: 
Let $n \geq 3$, $k=2^{n-1}$, and $d: \mathbb{N} \rightarrow \mathbb{R}$ is a non-invertible function.
The problem is to form a cyclic  sequence of patterns $C, C_0, C, C_1, C, C_2,\ldots,\allowbreak C, C_{k-1}$  where $C$ is a pattern of $n$ robots forming a ``circle'' occupied by $n-1$ robots with one in the center (see   Figure~\ref{Fig:CC} (a)), and $C_i$ (for $0 \leq i \leq k-1$) is  a configuration where the $n-1$ circle robots are in the exact same position, but the center robot occupies a point at distance $d(i)$ from the center on  the radius connecting the center to the missing robot position on the circle (see   Figure~\ref{Fig:CC}~(b)). In other words, the central robot moves to the designed position at distance $d(i)$ and comes back to the center, and the process repeats after the $2^{n-1}$  configurations $C_i$ have been formed.

\end{definition}

\begin{lemma}\label{lem:CYC-FSFimpo}
Let $n \geq 3$. $\exists \mathit{R} \in \mathcal{R}_n, \text{CYC} \not \in \FS^{F}(R)$.
This result holds even in presence of chirality and rigid movements.
\end{lemma}
\begin{proof}
Since the locations of the robots $r_1, \ldots, r_{n-1}$ are the same in the  $2^{n-1}$ configurations $C_i$ and $C$, $\FS$ robot $r_0$ in   configuration $C$ must determine which $C_i$ must be formed   only on the basis of its own color.  Since the number of available colors   is constant with respect to the number of robots $n$, $r_0$ cannot recognize, using its light, the correct distance $d(i)$ from the center of the circle.
\end{proof}

We now show that $\FC$ robots can solve $\text{CYC}$ under \SSY.   The proof is  by construction and the problem is solved by  Algorithm~\ref{algo:Cyclic-Cycles}. Intuitively, the $n-1$ robots on the circle act as a distributed counter using their lights to display the binary representation of the index $i$ of the next configuration to be formed. The increment of the counter is done appropriately changing the bits and maintaining the carry   like in a full-adder addition. 
Whenever activated, robot $r_0$ ``reads'' the information and understands when it is its time to move and what is its destination.

\Newcodeline
\begin{algorithm}[H]
\caption{{\tt Cyclic-Cycles} for robot $r_i$.}
\label{algo:Cyclic-Cycles}
\small
\begin{tabbing}
111 \= 11 \= 11 \= 11 \= 11 \= 11 \= 11 \= \kill
{\em Assumptions}\crm
\> \FC, \SSY\crm
\> \parbox[t]{\columnwidth-2em}{Let $r_0$ be located at the center of the circle and $r_1, \ldots, r_{n-1}$ be located on the circle in clockwise, and let $\CoS(r_i)=r_{i+1\bmod n}$ and $\CoP(r_i)= r_{i-1\bmod n}$. Robot $r_i$ has the following lights:}\crm
\>\> $\status \in \{ \cntr, \final\}$, initially set to $\cntr$;\crm
\>\>$b, c, \mathit{suc}_b$ $\in \{0,1\}$, initially set to $0,0,0$, respectively.
\crm
\crm
{\em Phase Look}\crm
\> Robot $r_i$ observes the positions of the other robots and their lights.\crm
\> Note that $r_i$ cannot see its own lights due to $\FC$.
\crm
\crm
{\em Phase Compute}\crm
\Cl {\bf if} $i=0$ {\bf then} \` // case for robot $r_0$ \crm
\Cl \> {\bf if} $\forall (1 \leq j \leq n-1) : (r_j.\status=\cntr)$\crm
\>\>\>\> {\bf and} {\bf not} ($r_0$ is at the final position) {\bf then}\crm
\Cl \>\> $r_0.\status \leftarrow \final$\crm
\Cl \>\> \parbox[t]{\columnwidth-4em}{$r_0.\mathit{des} \leftarrow$ the position at distance $d((r_{n-1}.b)(r_{n-2}.b)\ldots(r_1.b))$ from the center of the circle}  \\[1ex]
\Cl \>{\bf else if} $\forall (1 \leq j \leq n-1) : (r_j.\status=\final)$\crm
\>\>\>\> {\bf and}  {\bf not} ($r_0$ is at the center of the circle) {\bf then}\crm
\Cl \>\> $r_0.\status$ $\leftarrow \cntr$\crm
\Cl \>\> $r_0.b \leftarrow 0$\crm
\Cl \>\> $r_0.c \leftarrow 1$ \` // increment by one\crm
\Cl \>\> $r_0.\mathit{suc}_b$ $\leftarrow r_1.b$ \` // copy of $r_1.b$ \crm 
\Cl \>\> $r_0.\mathit{des} \leftarrow$ the center of the circle \crm

\Cl {\bf else} \` // case for robot $r_i$ where $i\not=0$\crm
\Cl \> {\bf if} $r_0$ is located at the final position\crm
\>\>\>\>  {\bf and} $r_0.\status=\final$ {\bf then} \crm
\Cl \>\> $r_i.\status \leftarrow \final$ \crm
\Cl \> {\bf else if} $r_0$ is located at the center of the circle  \crm
\>\>\>\> {\bf and} $\forall (0 \leq j \leq i-1) : (r_j.\status=\cntr)$  \crm
\>\>\>\> {\bf and} $\forall (i+1 \leq j \leq n-1) : (r_j.\status=\final)$ {\bf then} \crm
\Cl \>\> $r_i$.b $\leftarrow \CoP(r_{i}).c \oplus \CoP(r_{i}).\mathit{suc}_b$\crm
\Cl \>\> $r_i$.c $\leftarrow \CoP(r_{i}).c \cdot \CoP(r_{i}).\mathit{suc}_b$\crm
\Cl \>\> $r_i.\mathit{suc}_b \leftarrow \CoS(r_{i}).b$ \crm
\Cl \>\> $r_i.\status \leftarrow \cntr$ \crm
\crm

{\em Phase Move}\crm
\> Move to $r_{i}.\mathit{des}$
\end{tabbing}
\end{algorithm}

\begin{lemma}\label{lem:CYC-FCS}
Let $n \geq 3$.
 $\forall \mathit{R} \in \mathcal{R}_n$, $\text{CYC}  \in \FC^{S}(R)$, assuming chirality.
\end{lemma}
\begin{proof}
Each robot $r_i$ has 4 lights, $r_i.\status$, $r_i.b$, $r_i.c$, and $r_i.\suc_b$. Let
\[
\begin{aligned}
\Status(t) & = (r_0.\status(t), \ldots, r_{n-1}.\status(t))\,,\\ 
B_\mathit{value}(t) & = (r_0.b(t), r_1.b(t), \ldots, r_{n-1}.b(t))\,,\\ 
B_\mathit{carry}(t) & = (r_0.c(t), r_1.c(t), \ldots, r_{n-1}.c(t))\,, \text{ and }\\
B_\mathit{suc}(t) & = (r_0.\mathit{suc}_b(t), r_1.\mathit{suc}_b(t), \ldots, r_{n-1}.\mathit{suc}_b(t))\,.
\end{aligned}
\] 
The configuration $K(t)$ at time $t$ is defined by $(X(t), \Status(t),$  $B_\mathit{value}(t),$ $B_\mathit{carry}(t), B_\mathit{suc})$, where $X(t)$ is the set of locations of $n$ robots at time $t$.
$\Status(t)$ differentiates among  $C$, $C_i$ or intermediate states between $C$ and $C_i$.
$B_\mathit{value}(t)$ and $B_\mathit{carry}(t)$ encode the binary representation of $i$.
Finally, values in $B_\mathit{suc}(t)$ are copied from the lights $b$ of the successor robots along the circle, thus they can be used by the robots to derive the value of their light $b$ by observing the color of the predecessor's light $\suc_b$.

The initial colors of $\Status(0)$, $B_\mathit{value}(0)$, $B_\mathit{carry}(0)$, $B_\mathit{suc}(0)$ are $(\cntr, \ldots, \cntr)$,\\ $(0, 0, \ldots, 0)$, $(0, 0, \ldots, 0)$, $(0, 0, \ldots, 0)$,  and $(0,\allowbreak 0, \ldots, 0)$, respectively.
The initial locations of robots is shown in Figure~\ref{Fig:CC} (a).
Algorithm~\ref{algo:Cyclic-Cycles}  works as follows:
\begin{itemize}
\item When robot $r_0$ is activated at $t$, if $r_0$ observes
\[\Status(t)[1..n-1] = (\cntr,  \ldots,\cntr)\,,\]
and $r_0$ is located at the center of the circle (the pattern $C$ is formed),  $r_0$ changes $r_0.\status$ to $\final$ and moves to the ``final'' position computed by $d(B_{value}(t)[1..n-1])$. Note that since   rigidity is not assumed, as long as $r_0$ does not reach the final position $r_0$ keeps moving towards the final position.  When $r_0$ reaches the final position,   pattern $C_{p}$ is formed, where $p$ is the integer represented in binary by $B_{value}(t)[1..n-1]$.
\item When a robot $r_i$ (for $i \neq 0$) is activated at $t$, if $r_0$ is located at the final position determined by $d(B_{value}(t)[1..n-1])$,
$r_i$ changes $r_i.\status$ to $\final$. Note that although an $\FC$ robot $r_i$ cannot observe its own $r_i.b$, since $r_{i-1}.\suc_b$ holds $r_i.b$, $r_i$ can compute $d(B_{value}(t)[1..n-1])$. Then there is a time $t'>t$ 
such that $\Status(t')=(\final, \ldots, \final)$ while $C_{p}$ is not changed. 
\item When   robot $r_0$ is activated at $t$, if $r_0$ observes 
\[\Status(t)[1..n-1]=(\final,  \ldots ,\final)\,,\]
and the current pattern $C_p$ corresponds to $B_{value}(t)[1..n-1]$, $r_0$ prepares the formation of $C$. That is, it changes $r_0.\status$ to $\cntr$ and it sets $(r_0.b, r_0.c, r_0.\suc_b)$ to 
$(0,1,r_1.b)$, where $r_0.c$ means ``carry'' to the least significant bit,  and setting $r_0.c=1$ is a preparation for the increment of the binary value $B_{value}(t)[1..n-1]$. Then $r_0$   moves to the center of the circle.
\item When robot $r_0$ reaches the center of the circle and 
\[\Status(t)=(\cntr, \final,  \ldots ,\final)\,,\]
the increment of the binary value $B_{value}(t)[1..n-1]$ is performed sequentially, like in a full-adder addition. 
Firstly,  $r_1$ (when activated) changes $r_1.\status$ to $\cntr$ and it  computes the least significant bit ($r_1.b$) and the carry to the next bit ($r_1.c$) by using $r_0.c$ and $r_0.\suc_b(=r_1.b)$. Robot $r_1$ also copies $r_2.b$ into $r_1.\suc_b$, which will be used  so that $r_2$ recognizes its own $b$. In general, when $r_0$   is located at the center of the circle and
\[
\begin{split}
\qquad\Status(t)=(r_0.\status=\cntr, \ldots, r_{i-1}.\status=\cntr,\\
r_i.\status=\final, \ldots, r_{n-1}.\status=\final),
\end{split}
\]
robot $r_i$, upon activation, changes $r_i.\status$ to $\cntr$ and the computation of the increment ($r_{i}.b$ and $r_i.c$) by using $r_{i-1}.c$ and $r_{i-1}.\suc_b(=r_i.b)$ and copy of the next bit ($r_i.\suc_b$).
Lastly, $r_{n-1}$ changes $r_{n-1}.\status$ to $\cntr$ and computes the most significant bit. Then the increment is completed. 

\hspace{10pt} Note that when this configuration occurs and $r_i$ executes the algorithm, also $r_{i-1}$  may execute the algorithm, if activated (because it is a $\FC$ robot). However, even if $r_{i-1}$ executes the algorithm in the same round, the configuration stays unchanged and there is no effect on the increment of $r_i$. Then,  there is a time $t$ such that $\Status(t)=(\cntr, \ldots, \cntr)$ and $B_{value}(t)[1..n-1]$ has been increment and this configuration forms $C$ and the next movement of $r_0$ will be performed.

\hspace{10pt} It can be easily verified that this algorithm works in \SSY.
\end{itemize}
\end{proof}

The orthogonality of $\FS^{RS}$ and $\FC^{S}$ follows from Lemmas~\ref{lem:SRO-OBRS}, \ref{lem:SRO-FCSFSSimpo}, \ref{lem:CYC-FSFimpo}, and \ref{lem:CYC-FCS}, and the
  dominance of $\LU^{S}$ over $\FS^{RS}$ follows from Lemmas~\ref{lem:CYC-FSFimpo} and \ref{lem:CYC-FCS}. 


\begin{theorem}\label{th:RSvsS2}~
 \begin{enumerate}

  \item\label{FCRS<LUS} $\FS^{RS} < \LU^{S}$,
   \item\label{FCSneqFSRS} $\FS^{RS} \bot$  $\FC^{S}$.
\end{enumerate}
\end{theorem}

We show next the equivalence of $\LU^S$ and $\LU^{RS}$.
 
\subsection{Power of  \texorpdfstring{\RSY  in $\mathcal{LUMI}$}{}}\label{LUS=LURS}

We now show that, in spite of their power for $\FC$ and $\FS$,  \RSY robots  with full lights ($\LU$)
are {\em not} more powerful than \SSY robots with full lights. That is, 
$\mathcal{LUMI}^{S}$ is computationally equivalent to $\mathcal{LUMI}^{RS}$.

\begin{theorem}\label{th:sim-DS-by-S}
$\forall R \in \mathcal{R}, \LU^{RS}(R) \leq \LU^S(R)$.
\end{theorem}

We start by describing a simulation algorithm  ({\tt  sim-RS-by-S($A$)}) that works in  \SSY with full lights simulating an algorithm  $A$ working in \RSY with full lights.
Let algorithm $A$ use light with $\ell$ colors: $C=\{c_0,c_1, \ldots, c_{\ell-1}\}$. 
 A robot $r$ has available the following sets of  colors  for its full lights in {\tt  sim-RS-by-S($A$)}: 


%
%
 
\begin{enumerate}
\item $r.\light  \in C$ indicating its own light used in algorithm~$A$, initially set to  $r.\light=c_0$;
\item $r.\step  \in \{1,2,\all, \fst, \snd, m \}$ indicating the step currently under execution. Step~$i$ begins when $\forall \rho \in R \, (\rho.\step =i)$ for $i \in \{1,2,\all, \fst, \snd, m \}$.
Initially $r.\step$ is set to $1$ for all robots, and thus initially the simulation is in Step~$1$;
\item $r.\executed   \in \{\true , \false \}$ 
  indicating whether $r$ has executed algorithm~$A$ in the current mega-cycle (see below).
  Initially $r.\executed$ is set to $\false$ for all robots; and

\item $r.\charged \in \{ C, E, M \}$, where $C$, $E$, and $M$ stand for ``charged'', ``empty'', and ``just moved'' respectively.
The flag is used to ensure the validity of the simulated \RSY activation sequence, it indicates whether $r$ is charged and can execute the algorithm $A$.
Initially $r.\charged$ is set to $C$ for all robots.
\end{enumerate}

\begin{figure}
\centering
\includegraphics[page=2]{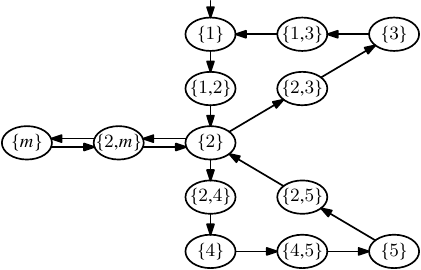}\\
    \caption{Transition Diagram of {\tt sim-RS-by-S($A$)}. Step 1: first phase; Step 2: second phase; Step 3: all robots moved in Step 1, reset $\charged$ and $\executed$ flags; Step 4: empty robots are charged; Step 5: robots moved are discharged; Step $\bm m$: reset executed flags.}
    \label{fig:transition-sim-RS-by-S}
\end{figure}

\Newcodeline
\begin{algorithm}[H]
\caption{{\tt sim-RS-by-S(\textit{A})}: predicates and subroutines for robot $r$.}
\label{algo:sim-RS-by-S}
\begin{tabbing}
111 \= 11 \= 11 \= 11 \= 11 \= 11 \= 11 \= \kill
%

\textbf{\em predicate} {\em all-robots-executed} \crm
\> $\forall \rho \, (\rho.\executed = True )$ \crm
\crm
\textbf{\em predicate} {\em all-robots-reset}\crm
\> $\forall \rho \, (\rho.\executed = \false )$ \crm
\crm

\textbf{\em subroutine} {\em Reset-All-Flags} \` // in Step~$\all$\crm
\> {\bf if} $\exists \rho \, (\rho.\executed = \true$ {\bf and} $\rho.\charged = E$) {\bf then}\crm
\>\> {\bf if} ($r.\executed = \true$ {\bf and} $r.\charged = E$)  {\bf then}\crm
\>\>\> $r.\executed \leftarrow \false$ \crm
\>\>\> $r.\charged \leftarrow C$ \crm
\> {\bf else}  $r.\step \leftarrow 1$ \crm
\crm

\textbf{\em subroutine} {\em Reset-Execution} \` // in Step~$m$\crm
\> {\bf if} {\bf not} $\forall \rho \, (\rho.\executed = \false)$ {\bf then}\crm
\>\>  {\bf if} $r.\executed \neq \false$ {\bf then} $r.\executed \leftarrow \false$ \crm
\>  {\bf else}  $r.\step \leftarrow 2$ \crm
\crm

\textbf{\em subroutine} {\em Charge-Empty-Flags} \` // in Step~$\fst$\crm
\> {\bf if} $\exists \rho \, (\rho.\charged = E)$ {\bf then}\crm
\>\> {\bf if} $r.\charged = E$ {\bf then} $r.\charged \leftarrow C$ \crm
\> {\bf else}  $r.\step \leftarrow \snd$ \crm
\crm

\textbf{\em subroutine} {\em Empty-Moved-Flags} \` // in Step~$\snd$\crm
\> {\bf if} $\exists \rho \, (\rho.\charged = M)$ {\bf then}\crm
\>\> {\bf if} $r.\charged = M$ {\bf then} $r.\charged \leftarrow E$ \crm
\> {\bf else}  $r.\step \leftarrow 2$
\end{tabbing}
\end{algorithm}

\Newcodeline
\begin{algorithm}[H]
\caption{{\tt sim-RS-by-S(\textit{A})}:  for robot $r$.}
\label{algo:sim-RS-by-main}
\begin{tabbing}
111 \= 11 \= 11 \= 11 \= 11 \= 11 \= 11 \= \kill
{\em Assumptions}:\crm
\> $\LU$, \SSY
\crm\crm
%
{\em State Look}\crm
\> \tabfill{Robot $r$ takes the snapshot of the configuration, including the colors of the lights of all the robots.}\crm

\crm
{\em State Compute}\crm
\Cl $r.\des \leftarrow (0,0)$ \` // $r$ does not move\crm\crm
\Cl {\bf if} $\forall \rho \, (\rho.\step =1)$ {\bf then} \` // Step 1 \crm
\Cl \> Set $r.\lcolor$ and $r.\des$ by the simulated algorithm \textit{A}\crm
\Cl \> $r.\executed \leftarrow \true$\crm
\Cl \> $r.\charged \leftarrow E$\crm
\Cl \> $r.\step \leftarrow 2$ \crm\crm

\Cl  {\bf else if} $\forall \rho \, (\rho.\step =2)$ {\bf then} \` // Step 2 \crm
\Cl \> {\bf if} $\forall \rho \, (\rho.\charged= E)$ {\bf then} $r.\step \leftarrow \all$ \` // mega-cycle continues
\crm
\Cl \>{\bf else if} {\em all-robots-executed} {\bf then} $r.\step \leftarrow m$ \` // mega-cycle ends\crm

\Cl {\bf else if}   $(r.\executed=\false)$ {\bf and} $(r.\charged=C)$ {\bf then}\crm

\Cl \>\> Set $r.\lcolor$ and $r.\des$ by the simulated algorithm \textit{A}\crm
\Cl \>\> $r.\executed \leftarrow \true$\crm
\Cl \>\> $r.\charged \leftarrow M$\crm
\Cl \>\> $r.\step \leftarrow \fst$ \crm\crm 

\Cl  {\bf else if}  $\forall  \rho \, (\rho.\step =\all$) {\bf then}  {\bf call} {\em Reset-All-Flags} \` // to Step~$1$\crm
\Cl {\bf else if}  $\forall \rho \, (\rho.\step =\fst)$ {\bf then} {\bf call} {\em Charge-Empty-Flags} \` // to Step~5\crm
\Cl {\bf else if}  $\forall \rho \, (\rho.\step =\snd$) {\bf then} {\bf call} {\em Empty-Moved-Flags} \` // to Step 2\crm
\Cl  {\bf else if}  $\forall  \rho \, (\rho.\step =m$) {\bf then}  {\bf call} {\em Reset-Execution} \` // to Step~$2$\crm
\crm

\Cl {\bf else if} $\forall \rho \, (\rho.\step=1$ {\bf or} $\rho.\step=2)$ {\bf then} $r.\step \leftarrow 2$ \crm 
\Cl  {\bf else if} $\forall \rho \, (\rho.\step=2$ {\bf or} $\rho.\step=\all)$  {\bf then} $r.\step \leftarrow \all$  \crm
\Cl  {\bf else if} $\forall \rho \, (\rho.\step=2$ {\bf or} $\rho.\step=\fst)$ {\bf then} $r.\step \leftarrow \fst$ \crm
\Cl  {\bf else if} $\forall \rho \, (\rho.\step=\fst$ {\bf or} $\rho.\step=\snd$) {\bf then} $r.\step \leftarrow \snd$ \crm 
\Cl  {\bf else if} $\forall \rho \, (\rho.\step=\snd$ {\bf or} $\rho.\step=2)$ {\bf then} $r.\step \leftarrow 2$ \crm 
\Cl  {\bf else if} $\forall \rho \, (\rho.\step=\all$ {\bf or} $\rho.\step=1)$ {\bf then} $r.\step \leftarrow 1$\crm

\Cl  {\bf else if} $\forall \rho \, (\rho.\step=2$ {\bf or} $\rho.\step=m)$ {\bf and} {\em all-robots-executed} {\bf then}\crm
\Cl \> $r.\step \leftarrow m$ \crm
\Cl  {\bf else if} $\forall \rho \, (\rho.\step=m$ {\bf or} $\rho.\step=2)$ {\bf and} {\em all-robots-reset} {\bf then} \crm
\Cl \> $r.\step \leftarrow 2$\crm\crm

{\em State Move}\crm
\> Move to $r.\des$
\end{tabbing}
\end{algorithm}

The simulating algorithm is presented in Algorithm~\ref{algo:sim-RS-by-S} (predicates and subroutines used in the algorithm) and \ref{algo:sim-RS-by-main} (the main algorithm), and Figure~\ref{fig:transition-sim-RS-by-S} shows the transition diagram as the robots change step's value. Since {\tt sim-RS-by-S($A$)}  works in \SSY, it must take care of  excluding  the prohibited patterns of  \RSY,  if they occur.

The simulation proceeds in two phases.
The first phase corresponds to the first $p\ge 0$ (where $p$ can be $\infty$) activations in a simulated \RSY activation schedule where all robots in $R$ are activated at each round.
The second phase corresponds to the remaining activation cycles where a strict subset of $R$ is activated at each round.
The robots execute $A$ in Step~$1$ in the first phase, and in Step~$2$ in the second phase.
The remaining steps serve for bookkeeping of flags $\executed$ and $\charged$.

The algorithm is a sequence of {\em mega-cycles}, each of which lasts the time it takes for  all robots to execute the simulated algorithm exactly once.
The cycle $\{1\}\rightarrow\{2\}\rightarrow\{3\}\rightarrow\{1\}$ of the transition diagram of {\tt sim-RS-by-S($A$)} corresponds to one mega-cycle in the first phase.
One mega-cycle in the second phase consists of several cycles $\{2\}\rightarrow\{4\}\rightarrow\{5\}\rightarrow\{2\}$ (until all robots have executed $A$) with one cycle $\{2\}\rightarrow\{m\}\rightarrow\{2\}$ to reset the flags and start the new mega-cycle.
Specifically, the states of the simulation are (see Figure~\ref{fig:transition-sim-RS-by-S}):
\begin{itemize}
\item State~$\{i,j\}$ (lines 19--28 of Algorithm~\ref{algo:sim-RS-by-main}) are the transition states between Step $i$ and Step $j$. Activated robots do not execute $A$ nor change the values of the $\charged$ and $\executed$ flags, but only set their $r.\step$ to either $i$ or to $j$ (depending on the specific case).

\item Step~$1$ (lines 2--6) executes one cycle of the simulated algorithm $A$.
Note that when Step~$1$ begins execution, it holds that $\forall \rho \in R \, (\rho.\executed = \false$ and  $\rho.\charged = C)$.

\item Step~$2$ (lines 7--14) checks if all robots have their $\charged$ flag set to $E$, that is, if all robots were activated in Step~$1$.
If so, the first condition of \RSY is satisfied, and the simulation proceeds to Step~$\all$ resetting all the $\charged$ and $\executed$ flags to prepare for the next move in Step~$1$.
If there are robots with their $\charged$ flag set to $C$, there are two cases:
    
(1)~If all robots have their $\executed$ flag set to $\true$, the mega-cycle has finished.
In this case, the simulation proceeds to Step~$m$ resetting all the executed flags.
When Step~$m$ ends, it returns to Step~$2$.
    
(2)~If some robots have their $\executed$ flag set to $\false$, the mega-cycle has not finished yet.
In this case, the simulation proceeds executing $A$.
As soon as among the activated robots there is a non-empty subset $R'$ with $\charged=C$ and $\executed=\false$, they execute algorithm $A$, set $\charged$ to $M$ and $\executed$ to $\true$, and proceed to Step~$4$.
Afterwards, the remaining robots proceed to Step~$4$ without executing $A$ or changing their $\charged$ and $\executed$ flags.

\item Step~$\all$ (line 15) resets all the $\charged$ and $\executed$ flags to prepare the next move in Step~$1$.
Note that as long as Step~$\all$ is executed, full activation continues.

\item Step~$\fst$ (line 16) updates the $\charged$ flag of the robots that did not execute $A$ in preceding Step~$2$ from $E$ to $C$ (the robots that did not execute are recharged).
    
\item Step~$\snd$ (line 17) updates the $\charged$ flag of the robots that executed $A$ in preceding Step~$2$ from $M$ to $E$ (the robots that executed are discharged).
    
\item Step~$m$ (line 18) executes when the mega-cycle is completed at the beginning of Step~$2$.
It resets the $\executed$ flags of all robots and goes back to Step~$2$.
Note that Step~$m$ does not affect the $\charged$ flag, thus the robots which executed $A$ in the last activation cycle of the preceding mega-cycle are discharged in the first activation cycle of the new mega-cycle and cannot be activated. 
\end{itemize}

The initial configuration of {\tt  sim-RS-by-S($A$)}  satisfies  $r.\step=1$, $r.\charged=  C$ and $r.\executed =\false$ for any robot $r$.
For a configuration $K$, if the set of values of $r.\step$ appearing in $K$ is $V$, we say that the {\em step configuration} of $K$ is $V$.

Observe that at any moment in time of {\tt  sim-RS-by-S($A$)} the step configuration $V\in \{\{1\}, \{2\}, \{3\}, \{4\}, \{5\}, \{m\}, \{1,2\}, \{2,3\}, \{3,1\},\allowbreak \{2,4\}, \{4,5\}, \{2,5\}, \{2,m\} \}$.
Indeed, the case when $|V|=1$ corresponds to a (sub)set of activated robots performing the same action (executing $A$, or modifying the flags $\charged$ and $\executed$).
These robots may update their $\step$ value leading to a new step configuration $V'$ of a size at most two.
In the case when $|V|=2$ no actions are performed by the robots except of updating their $\step$ to some value $i\in V$ (remaining robots transition to Step~$i$), eventually leading to the next step configuration $V'=\{i\}$.

To prove the correctness of {\tt  sim-RS-by-S($A$)} we start with proving the following properties of the step configurations.
\begin{lemma}~
\vspace{-\topsep}
\begin{description}
\item[P$1$] In Step~$1$, all the robots have $(\charged,\executed)=(C,\false)$.\item[P$2$] In Step~$2$, (1)~all the robots have flags $(\charged,\executed)\in$\\ $\{(C,\false),(C,\true),(E,\true)\}$, (2)~a non-empty subset of robots $R'$ have $(\charged,\allowbreak\executed)=(E,\true)$, and (3)~if $|R'|=|R|$ then the preceding Step was~$1$.
\item[P$3$] In Step~$3$, for all the robots $(\charged,\executed)\in\{(C,\false),$ \\ $(E,\true)\}$.
\item[P$4$] In Step~$4$, (1)~none of the robots have $(\charged,\executed) \in \{(M,\false),(E,\false)\}$, (2)~a non-empty subset of robots $R_m$ have flags $(\charged,\executed)=(M,\true)$ with $|R_m|<|R|$; these and only these robots executed a round of algorithm $A$ in the preceding Step~$2$, and (3)~a non-empty subset of robots $R_x$ have flags $(\charged,\executed)\in\{(E,\true),(C,\true)\}$ with $|R_x|<|R|$.
\item[P$5$] In Step~$5$, (1)~none of the robots have $(\charged,\executed) \in \{(M,\false),\allowbreak(E,\false)\}$, (2)~a non-empty subset of robots $R_m$ have flags $(\charged,\allowbreak\executed)\in\{(M,\true),(E,\true)\}$ with $|R_m|<|R|$; these and only these robots executed a round of algorithm $A$ in the preceding Step~$2$, and (3)~a non-empty subset of robots $R_x$ have flags $(\charged,\allowbreak\executed)=(C,\true)$ with $|R_x|<|R|$.
\item[P$m$] In Step~$m$, (1)~none of the robots have their flag $\charged=M$, and (2) a non-empty subset of robots $R''$ have $\charged=C$ with $|R''|<|R|$.
\end{description}
\end{lemma}
\begin{proof}
Initially property P$1$ holds. Below we show that for any transition from Step~$i$ to Step~$j$ of {\tt sim-RS-by-S($A$)}, if the corresponding property P$i$ held before the transition, the property P$j$ holds after the transition.
Thus, by induction the lemma holds.

$\text{P}1\rightarrow\text{P}2$: If property P$1$ holds in Step~$1$, in the succeeding Step~$2$ property P$2$ also holds.
Indeed, when transiting from Step~$1$ to Step~$2$, the robots either set their flags to $\charged = E$ and $\executed=\true$, or do not update these flags at all.
And a non-empty set of robots does the former.
The last condition trivially holds.
Thus P$2$ holds.

$\text{P}2\rightarrow\text{P}2$: If after activation of a set of robots $R'$ in Step~$2$ the step configuration remains $\{2\}$, then (1) not all robots in $R$ had their flag $\charged=E$, (2) not all robots in $R$ had their flag $\executed=\true$, and (3) none of the robots in $R'$ had their flags $\charged=C$ and $\executed = \false$.
Thus, no actions are taken, and no flags are changed; P$2$ holds.

$\text{P}2\rightarrow\text{P}3$: If property P$2$ holds in Step~$2$ and the succeeding Step is~$3$, then in it property P$3$ also holds.
From Step~$2$ robots can transition to Step~$3$ only if all the robots have $\charged=E$.
Thus, there were only robots with $\charged = E$ and $\executed=\true$ in Step~$2$.
They proceed to Step~$3$ without changing their flags.
Thus P$3$ holds.

$\text{P}2\rightarrow\text{P}4$: If property P$2$ holds in Step~$2$ and the succeeding Step is~$4$, then in it property P$4$ also holds.
From Step~$2$ robots can transition to Step~$4$ only if among the activated robots there is a non-empty subset with $(\charged,\executed)=(C,\false)$.
These robots update their flags to $(\charged,\executed)=(M,\true)$, and proceed to Step~$4$.
The remaining robots proceed to Step~$4$ with no changes to the $\charged$ and $\executed$ flags.
Furthermore, there were no robots in Step~$2$ with $\charged=M$ or $(\charged,\executed)=(E,\false)$, thus P$4$ holds.

$\text{P}3\rightarrow\text{P}3$: In Step~$3$ robots with $(\charged,\executed) = (E,\true)$ update these flags to $(C,\false)$. Thus P$3$ continues to hold.

$\text{P}3\rightarrow\text{P}1$: If property P$3$ holds in Step~$3$ and the succeeding Step is~$1$, in it property P$1$ also holds.
Robots only transition from Step~$3$ to Step~$1$ when none of the robots have $(\charged, \executed) = (E,\true)$.
Thus P$1$ holds.

$\text{P}4\rightarrow\text{P}4$: In Step~$4$ robots with $\charged = E$ update the value of the flag to $C$.
Thus P$4$ continues to hold.

$\text{P}4\rightarrow\text{P}5$: If property P$4$ holds in Step~$4$ and the succeeding Step is~$5$, then in it property P$5$ also holds.
Robots only transition from Step~$4$ to Step~$5$ when none of the robots have $\charged = E$.
Thus all three conditions of P$5$ hold.

$\text{P}5\rightarrow\text{P}5$: In Step~$5$ robots with $\charged = M$ update their flag to $E$.
Thus P$5$ continues to hold.

$\text{P}5\rightarrow\text{P}2$:  If property P$5$ holds in Step~$5$, in the succeeding Step~$2$ property P$2$ also holds.
Robots only transition from Step~$5$ to Step~$2$ when none of the robots have $\charged = M$.
The last condition trivially holds.
Thus P$2$ holds.

$\text{P}2\rightarrow\text{P}m$: If property P$2$ holds in Step~$2$ and the succeeding Step is~$m$, then in it property P$m$ also holds.
From Step~$2$ robots can transition to Step~$m$ only if all the robots have $\executed=\true$ and not all the robots have $\charged=E$.
In this case none of the robots update their flags $\charged$ and $\executed$, and both conditions of P$3$ hold.

$\text{P}m\rightarrow\text{P}m$: In Step~$m$ robots with $\executed = \false$ update their flag to $\true$.
Thus P$m$ continues to hold.
\end{proof}

As mentioned above, {\tt sim-RS-by-S(A)} executes Step~$1$ in the first phase when all the robots are activated at every simulated activation cycle of \RSY.
Each mega-cycle then consists of one execution of Step~$1$.
As soon as a subset of robots is activated in Step~$1$, a second phase begins.
In it, each mega-cycle consists of multiple executions of Step~$2$.
From property P$2$ it follows that at every simulated activation cycle of \RSY a non-empty subset of robots is activated.
From properties P$2$, P$4$, and P$5$, it follows that when a robot executes algorithm $A$ in a mega-cycle, it will not execute in the same mega-cycle again.
The mega-cycle lasts until all the robots have executed $A$.
Furthermore, from properties P$2$, P$4$, P$5$, and P$m$, it follows that in the beginning of the new mega-cycle in Step~$2$, the robots that executed the algorithm in the last simulated activation cycle of \RSY have their flag $\charged=E$, and otherwise the robots have $\charged=C$.
Thus, simulated activation cycle of \RSY is valid and is fair.

Note that, if   algorithm~$A$ uses $\ell$ colors, the simulating algorithm~{\tt sim-RS-by-S(A)} uses $O(\ell)$ colors.

\begin{lemma}
Algorithm {\tt sim-RS-by-S(A)} correctly simulates execution of algorithm $A$ run on $\LU^\RS$ robots by $\LU^S$ robots. 
\end{lemma}

We have then obtained Theorem~\ref{th:sim-DS-by-S} and since the reverse relation is trivial, the next theorem holds.


\begin{theorem}\label{th:LUMISeqLUMIRS}
$\mathcal{LUMI}^{S} \equiv$  $\mathcal{LUMI}^{RS}$.
\end{theorem}


 \section{Relationship between \FSY\ and \RSY}

We have seen that \RSY\ is more powerful than  \SSY  in  \FC, \FS,  \OB, 
and it has the same computational power in \LU.
 To better understand  the power of  \RSY\ among the classical synchronous schedulers,
 we now turn our attention to  the  relationship between \RSY\ and \FSY.
 
\subsection{Dominance of \FSY\ over \RSY}

The problem Center of Gravity Expansion (CGE) was used to show dominance of \FSY over \SSY, that is, CGE is solvable in $\FC^{F}$  and $\FS^{F}$ but is not solvable in $\LU^{S}$ \cite{FSW19}. This problem can also be used to show dominance of \FSY\ over \RSY.  In fact, we can obtain a stronger result  showing that CGE is not solvable in $\LU^{F'}$, where $F'$ is any scheduler such that the first activation does not contain all robots.

 Let $P = \{(x_1, y_1),(x_2, y_2), \ldots, (x_n, y_n)\}$ be the set of positions occupied by
robots $R =\{ r_1, r_2,\ldots,  r_n\}$ and let $c = (c_x, c_y)$ be the coordinates of the CoG (Center of Gravity)
of $P$ at time $t = 0$. All these values are expressed according to a global
coordinate system not available to the robots, which have their own local
ones.
The following problem prescribes the robots to perform a sort of ``expansion'' of the initial configuration with respect to their center of gravity;
specifically, each robot must move away from $(c_x, c_y)$ to the closest 
position with lattice point corresponding to doubling its distance from it. More precisely:

    \begin{definition} {\bf CENTER OF GRAVITY EXPANSION (CGE)}:
Let $R$ be a set of robots with $|R|\geq 2$. The CGE problem requires each robot $r_i \in R$ to move 
from its initial position $(x_i,y_i)$ directly to $(f(x_i,c_x),f(y_i,c_y))$ where $f(a,b)=\lfloor 2a-b \rfloor$, 
away from the Center of Gravity of the initial configuration so that each robot doubles its initial distance from it and no longer moves. 
\end{definition}

Note that, if $f(a, b) = 2a - b$, the result would be a perfect expansion of
the initial configuration, doubling the distance of each robot from the center
of gravity; this expansion would clearly preserve the center of gravity. The
presence of the floor makes the configuration expand in a more irregular
fashion, resulting in a new configuration with a different center of gravity.
Also note that each robot $r_i$ perceives itself as the center of its own
coordinate system. Expressed in $r_i$'s local coordinate system, the center
of gravity is located at point $(c_{i,x}, c_{i,y}) = (c_x - x_i
, c_y - y_i)$. Applying the
appropriate translations, upon looking, the destination of $r_i$
is then given by
the local function $f_i(c_{i,x}, c_{i,y})) = (\lfloor -c_{i,x}\rfloor, \lfloor-c_{i,y}\rfloor)$ (which is the equivalent
of the global $f$ indicated in the problem definition). For ease of discussion,
in the following we use the global coordinates.

\begin{lemma}\label{short-lem:infinite}
For a set of points $P =\{(x_i,y_i)| 1 \leq i \leq n \}$ there exist infinite sets of points $Q_i(1 \leq i \leq n)$
such that $(x_i,y_i) \in Q_i(1 \leq i \leq n)$ and $\forall (x_i',y_i') \in Q_i(1 \leq i \leq n)$:
$f(x_i, c_x) = f(x', c_x')$, and  $f(y_i, c_y) = f(y', c_y')$, where $(c_x', c_y')$ is the coordinates of the CoG of $\{(x_i', y_i')| 1 \leq i \leq n\}$.
\end{lemma}
%
%

\begin{lemma}
Let $n \geq 2$. $\exists \mathit{R} \in \mathcal{R}_n$: $\text{CGE} \not \in \LU^{F'}(R)$,
 where $F'$ is any scheduler such that  the first activation does not contain all robots.  
\end{lemma}
\begin{proof}
Let $k$ be the number of activated robots in the first round. Note that $1 \leq k< n$ by assumption.
By contradiction, suppose that the problem is solvable from an arbitrary initial configuration. Consider two scenarios with $R=\{r_i| 1\leq i \leq n\}$.

In scenario $A$ the global coordinates of the robots are:
$r_i = (c_i,0)(1 \leq i \leq n)$, where $c_i(1 \leq i \leq n-1)$ are different negative integers, $c_n$ is a natural number, and $\Sigma_{i=1}^{n}c_i=0$; thus, CoG is at $(0,0)$.

In scenario $B$ the global coordinates of the robots are:
$r_i = (c_i+\epsilon,0)(1 \leq i \leq n-1), r_{n}=(c_n,0)$, where $c_i(1 \leq i \leq n-1)$ are different negative integers, $c_n$ is a natural number, $\Sigma_{i=1}^{n}c_i=0$ and $1>\epsilon>0$; thus CoG is at $((n-1)\epsilon/n,0)$. 

Consider an execution ${\mathcal E}$ where the scheduler activates $r_1, r_2,\ldots,r_k$ at time $t_0$, where $1 \leq k \leq n-1$ by assumption. In this execution, in scenario $A$, $r_i(1 \leq i \leq k)$ moves to $f((c_i,0),(0,0))=(2c_i,0)$, possibly changing color; the new configuration at time $t_0+1$ is $r_i=(2c_i,0)(1 \leq i \leq k), r_i=(c_i+\epsilon,0)(k+1 \leq i \leq n-1)$, and $r_{n}=(c_n,0)$.
 Observe that the same configuration would have been obtained by the same execution also in scenario $B$ where CoG would be at $((n-1)\epsilon/n,0)$ and $r_i(1 \leq i \leq k)$ would have moved to $f(c_i+\epsilon,(n-1)\epsilon/n),(0,0))=(c_i,0)$ (if $(n+1)\epsilon/n <1$) as well.
 
 In execution ${\mathcal E}$, let the scheduler activate $R_{t_0+1}$ (it can be $R$) at time $t_0+1$ and let $r_n \in R_{t_0+1}$.
 Robot $r_n$ may know that $r_i(1 \leq i \leq k)$ have already reached their destination by looking at their colors; that is, they may know that their positions have been calculated according to the target function $f$.  However, by Lemma~\ref{short-lem:infinite}, we know that $f^{-1}$ corresponds to an infinite set of coordinates which contains, in particular, both $(c_i,0)$ and $(c_i+\epsilon,0)(1 \leq i \leq k)$.
 
 In scenario $A$ the unique correct destination for $r_n$ would be $f((c_n,0),(0,0))=(2c_n,0)$. In scenario $B$, on the other hand, to solve the problem, $r_n$ should move to $f((c_n, (n-1)\epsilon/n),(0,0))=(2c_n-1,0)$. 
Since robot $r_n$  observes the same configuration at time $t_0+1$, $r_n$ cannot distinguish the two scenarios.
We have a contradiction. 
\end{proof}

Since \RSY\ contains patterns in $F'$, we have the following.

\begin{corollary}\label{co:CGE}
Let $n \geq 2$. $\exists \mathit{R} \in \mathcal{R}_n$, $\text{CGE} \not \in \LU^{RS}(R)$.
\end{corollary}  
As a consequence, we have dominance of \FSY\ over \RSY\ as follows;

\begin{theorem}\label{th:FvsRS} 
 $\forall \mathcal{X}\in\{\OB,\FC,\FS,\LU\}$\\   $\mathcal{X}^{F} > \mathcal{X}^{RS}$.
\end{theorem}

\subsection{Orthogonality of \FSY\ with \RSY}
 
In this section, we show the orthogonality of $\FS^{F}$ with $\mathcal{X}^{RS}$ for $\mathcal{X} \in \{\LU, \FC \}$, and of $\OB^{F}$ with $\mathcal{X}^{RS}$ for $\mathcal{X} \in \{\LU, \FC, \FS\}$. 
The orthogonality of $\FS^{F}$ with $\mathcal{X}^{RS}$, for $\mathcal{X} \in \{\LU, \FC \}$, can be obtained by
(1) CYC is not in $\FS^{F}$ (Lemma~\ref{lem:CYC-FSFimpo}) and is in $\FC^{S}$ (Lemma~\ref{lem:CYC-FCS}), and thus CYC is in $\FC^{RS}$ and $\LU^{RS}$, and
(2) CGE is in $\FS^{F}$ (\cite{FSW19}) and is not in $\LU^{RS}$ (Corollary~\ref{co:CGE}),  and thus CGE is not in $\FC^{RS}$ .

\begin{theorem}\label{th:FSFneqLURSFCRS}~
\begin{enumerate}
\item \label{FSFneqLURS}$\FS^{F} \bot$ $\LU^{RS}$,
\item \label{FSFneqFCRS}$\FS^{F} \bot$ $\FC^{RS}$.
\end{enumerate}
\end{theorem}

The problems OSP and CGE* were used to show $\OB^{F} \bot$ $\LU^{S}$, that is,
problem OSP can be solved in $\LU^{S}$, but not in $\OB^{F}$, and
problem CoG* can be trivially solved  in $\OB^{F}$, but not in  $\LU^S$~\cite{FSW19}.



These problems can be used to show $\OB^{F} \bot$ $\mathcal{X}^{RS}$ ($\mathcal{X} \in \{\LU, \FC, \FS \}$) by using the following lemma \ref{lem:OSP}, Corollary~\ref{co:CGE},   and the fact that CGE* can be solved in $\FS^{F}$~\cite{FSW19}.




 \begin{lemma}\label{lem:OSP}~
\begin{itemize}
    \item $\exists \mathit{R} \in \mathcal{R}_2$, $\text{OSP} \not \in \OB^{F}(R) $~{\em \cite{DFPSY}},
    \item $\forall \mathit{R} \in \mathcal{R}_2$, $\text{OSP}  \in \FS^{RS}(R) \cap \FC^{RS}(R)$.
\end{itemize}
\end{lemma}

It follows that
\begin{theorem}\label{th:orthFSOBwithRS}~
\begin{enumerate}
    \item\label{OBFneqLURS} $\OB^{F} \bot$ $\LU^{RS}$
    \item\label{OBFneqFCRS} $\OB^{F} \bot$ $\FC^{RS}$    
    \item\label{OBFneqFSRS}  $\OB^{F} \bot$ $\FS^{RS}$
\end{enumerate}
\end{theorem}




\section{Power within  \RSY}\label{LUS=FCRS}

In this section we determine the relationships among  $\LU^{RS},\\ \FC^{RS}$, and $\FS^{RS}$, and we show  that:\\
 $\FC^{RS} = \LU^{RS}  > \FS^{RS}$.

\begin{theorem}\label{th:sim-LUMI-by-FCOM}
$\forall R \in \mathcal{R}, \LU^{S}(R) \subseteq \FC^{RS}(R)$.
\end{theorem}


\Newcodeline
\begin{algorithm}[H]
\caption{{\tt sim-LUMI-by-FCOM(\textit{A})}: predicates and subroutines for robot $r$ at location $x$.}
\label{algo:simA}
\begin{minipage}{\columnwidth}
\medskip
\textbf{\em predicate}  {\em checked-all-neighbors-and-me}\\
\-\qquad $\forall \rho \neq r \ (\rho.\checked = \true$ {\bf and} $\rho.\CoS.\checked = \true)$

\bigskip

\textbf{\em predicate} {\em own-executed}\\
\-\qquad $\CoP(x).\executed \setminus \executed(x) = \{\true\}$,\\
\-\qquad\qquad \parbox{\linewidth - 4em}{where $\executed(x)$ is the set of $\executed$ colors seen by $r$ at its own location $x$ (note that, by definition, this set does not include $r$'s own color)}

\bigskip

\textbf{\em predicate}  {\em all-robots-executed}\\
\-\qquad $\forall \rho \neq r \ (\rho.\executed = \true) \textbf{ and } r.\executed$

\bigskip

\textbf{\em predicate}  {\em all-robots-reset} \\
\-\qquad $\forall \rho \neq r \ (\rho.\executed = \false)$

\bigskip

\textbf{\em predicate} {\em checked-flags-reset($\rho$)} \\
\-\qquad $\rho.\CoS.\lcolor = \emptyset$ {\bf and} $\rho.\CoS.\executed = \{\false\}$\\
\-\qquad\qquad {\bf and } $\rho.\CoS.\checked = \false$ {\bf and} $\rho.\checked = \false$

\bigskip

\textbf{\em subroutine} {\em Copy-Colors-of-Neighbors}\\
\-\qquad $r.\CoS.\lcolor \leftarrow \CoS(x).\lcolor$

\bigskip

\textbf{\em subroutine} {\em Copy-Execution-of-Neighbors}\\
\-\qquad $r.\CoS.\executed \leftarrow \CoS(x).\executed$

\bigskip

\textbf{\em subroutine} {\em Determine-own-color}\\
\-\qquad $r.\lcolor \leftarrow   \CoP(x).\CoS.\lcolor \setminus \lcolor(x)$,\\
\-\qquad\qquad \parbox{\columnwidth-4em}{where $\lcolor(x)$ is the set of $\lcolor$ colors seen by $r$ at its own location $x$ (note that, by definition, this set does not include $r$'s own color)}

\bigskip

\textbf{\em subroutine} {\em Reset-Checking-Flags}\\
\-\qquad $r.\CoS.\lcolor \leftarrow \emptyset$ \\
\-\qquad $r.\CoS.\executed \leftarrow \{ \false \}$ \\
\-\qquad $r.\CoS.\checked \leftarrow$ $\false$ \\
\-\qquad $r.\checked \leftarrow$ $\false$
\medskip
\end{minipage}
\end{algorithm}

\begin{algorithm}[tbp]
\caption{{\tt sim-LUMI-by-FCOM(\textit{A})}: for robot $r$ at location $x$.}
\label{algo:simAmain}
\begin{algorithmic}[1]
\small
\Statex \hspace{-2em}{\em Assumptions}: Let $x_0, x_1, \ldots, x_{m-1}$ be the circular arrangement on the configuration $C$, and let $\CoS(x_i)=x_{i+1\bmod m}$ and $\CoP(x_i)= x_{i-1\bmod m}$.

\bigskip

\Statex	\hspace{-2em}{\em State Look}
\Statex Observe, in particular, ${\tt pred}(x).\step$, ${\tt pred}(x).\lcolor(\executed)$, ${\tt suc}(x).\lcolor (\executed)$, as well as $r.\lcolor.\here$ (note that, for this, $r$ cannot see its own color).

\bigskip

\Statex	\hspace{-2em}{\em State Compute}
\State $r.\des \leftarrow r.\pos$
	
\medskip

\If{$\forall \rho \neq \ r(\rho.\step =1)$} \hfill // step 1 (CopyColors \& Execution)
\State {\bf call} {\em Copy-Colors-of-Neighbors}
\State {\bf call} {\em Copy-Execution-of-Neighbors}
\If{$\forall \rho  \in \CoS(x) (\rho.\checked =\true)$} $r.\CoS.\checked \leftarrow \true$
\EndIf
\State $r.\checked \leftarrow \true$
\If{{\bf not} {\em checked-all-neighbors-and-me}}
\State $r.\step \leftarrow  1$
\Else \State $r.\step \leftarrow  2$
\EndIf

\medskip

\ElsIf{$\forall \rho  \neq r \ (\rho.\step=\PS)$} \hfill // start step \PS (PerformSimulation)
\If{{\em all-robots-executed}} $r.\step \leftarrow \RA$ \hfill // if mega-cycle ends
\ElsIf{{\em own-executed}} $r.\step \leftarrow \PS$
\Else
\State{\bf call} {\em Determine-own-color}
\State Execute the {\tt Compute} of $A$ \hfill // with my color $r.\lcolor$, determining destination $r.\des$
\State $r.\executed \leftarrow \true$
\State $r.\step \leftarrow \RC$
\EndIf

\medskip

\ElsIf{$\forall \rho \neq r \ (\rho.\step=\RC)$} \hfill // step \RC (ResetCheckingFlags)
\If{$\forall \rho \neq r \, (${\em checked-flags-reset$(\rho))$}}
\State $r.\step \leftarrow \CCA$
\State {\bf call} {\em Reset-Checking-Flags}
\ElsIf{$\exists \rho \neq r \, ($ {\bf not} {\em checked-flags-reset$(\rho))$})}
\State $r.\step \leftarrow \RC $
\State {\bf call} {\em Reset-Checking-Flags}
\EndIf

\medskip

\ElsIf{$\forall \rho \neq r \, (\rho.\step=\RA)$} \hfill // step \RA (ResetExecutionFlags)
\If{{\bf not} $\forall \rho \neq r \, ((\rho.\executed = \false)$\\\qquad\qquad {\bf and}  $(\rho.\suc.\executed =\{\false\}))$}
\State $r.\step \leftarrow \RA$
\State $r.\executed \leftarrow \false$
\State $r.\CoS.\executed  \leftarrow \{\false\}$
\Else
\State $r.\step \leftarrow \PS$
\State $r.\executed \leftarrow \false$
\State $r.\CoS.\executed  \leftarrow \{\false\}$
\EndIf

\medskip

\ElsIf{$\forall \rho  \neq r \, ((\rho.\step=\CCA)$  {\bf or} ($\rho.\step=\PS))$} $r.\step \leftarrow \PS$
\ElsIf{$\forall \rho  \neq r \, ((\rho.\step=\PS)$  {\bf or} ($\rho.\step=\RC))$} $r.\step \leftarrow \RC$ 
\ElsIf{$\forall \rho \neq r \, ((\rho.\step=\PS)$ {\bf or} ($\rho.\step=\RA))$ 
{\bf and}  {\em all-robots-executed}} $r.\step \leftarrow \RA$
\ElsIf{$\forall \rho \neq r \, ((\rho.\step=\RC)$ {\bf or} ($\rho.\step=\CCA))$} $r.\step \leftarrow \CCA$
\ElsIf{$\forall \rho \neq r \, ((\rho.\step=\RA)$ {\bf or} ($\rho.\step=\PS))$ 
{\bf and} {\em all-robots-reset}} $r.\step \leftarrow \PS$
\EndIf

\medskip

\Statex \hspace{-2em} {\em State Move}
\Statex Move to $r.\des$
\end{algorithmic}
\end{algorithm}

Let \textit{A} be an algorithm in \SSY\ by $\mathcal{LUMI}$ with Disorientation and non-rigid movement,
and let \textit{A} use light with $\ell$ colors: $C=\{c_0, c_1, \ldots c_{\ell-1} \}$.
We now extend the simulation algorithm  of $\LU$ robots by $\FC$ robots in \FSY\ described in \cite{FSW19},
designing a more complex simulation algorithm  {\tt sim-LUMI-by-FCOM(\textit{A})}
 (Algorithms~\ref{algo:simA} and~\ref{algo:simAmain}) by $\mathcal{FCOM}$ robots in \RSY. 
The main ideas will stay the same.
Robots first copy the lights of their neighbors.
This in turn allows them to look at their neighbors to gain information about the color of their own light.
Next, some robots activate and have enough information to execute a step in \textit{A}.
Lastly, the robots reset their states such that they can start copying the lights of their neighbors again.
This cycle repeats itself and every time some robots execute algorithm \textit{A}.
Contrary to the simulation algorithm in \FSY, we need to take extra care here to ensure that the resulting schedule is fair and every robot executes a step in \textit{A} infinitely often.
After all, the same subset of robots might execute algorithm \textit{A} each time the robots perform this cycle.
As with {\tt sim-RS-by-S(\textit{A})}, we consider the concept of {\em mega-cycles}, in which every robot executes a step in \textit{A} exactly once.
To facilitate this, each robot gets a flag indicating if the robot has executed \textit{A} already this mega-cycle.
To give a robot information about its own execute status, we let the robots copy this flag in the same way as they copy the other lights of their neighbors.
As soon as all robots have executed \textit{A} in a certain mega-cycle, these flags get reset and the next mega-cycle starts.

To facilitate copying the lights of neighbors, we define a circular ordering on the locations of the robots. When chirality is assumed, this defines for each location $x$ a successor location $\CoS(x)$ and a predecessor location $\CoP(x)$. The assumption of chirality will later be lifted.

%

Figure~\ref{fig:transition-sim-LUMI-by-FCOM} shows the transition diagram as the robots change steps' values.


Robot $r$ at location $x$ has available the following colors of $\mathcal{FCOM}$.

\begin{figure}
\centering
\includegraphics{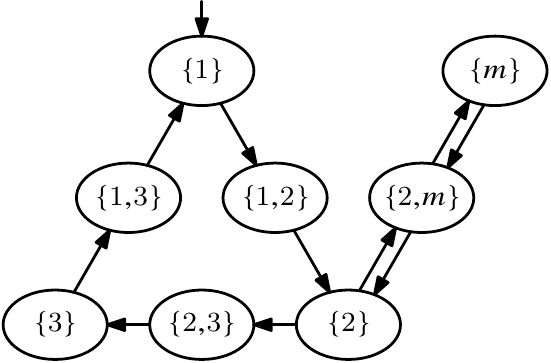}
    \caption{Transition Diagram of {\tt sim-LUMI-by-FCOM(\textit{A})}. Step~1: copy colors and activation; step~2: check if the mega-cycle ends; step~3: reset check flags; step~$\bm m$: reset activation.}
    \label{fig:transition-sim-LUMI-by-FCOM}
\end{figure}

\begin{enumerate}
\item $r.\light \in C$ indicating its own light used in algorithm~$\textit{A}$, initially set to    $r.\light=c_0$,

\item $r.\CoS.\light$ indicating a set of colors from $2^{C}$ at \CoS(x), initially set to  $r.\CoS.\light= \emptyset$,
\item $r.\step  \in \{1,2,3,\RA\}$ indicating the step of the algorithm, where step~$1,2,3$ and $\RA$ are CopyColors\&Execution, PerformSimulation,  ResetCheckingFlags and ResetExecutionFlags, respectively, initially set to $r.phase=  1$,

\item $r.\executed  \in \{\true, \false \}$ indicating whether $r$ has executed the algorithm~$A$ or not in a mega-cycle, initially set to $r.\executed=\false$. 
\item $r.\CoS.\executed$ indicating a set from $2^{\{\true, \false\}}$ at   \CoS(x) and having the same role of $r.\CoS.\light$, initially set to
$ r.\CoS.\executed=\emptyset$,
\item$r.me.checked$, $r.\CoS.checked$, $\in \{\true, \false \}$ indicating  whether all `light's and `executed's  are correctly set, initially set to $r.me.checked$,   $r.\CoS.checked=\false$.
\end{enumerate}

The algorithm is a sequence of mega-cycles, each of which lasts until all robots execute the simulated algorithm exactly once,  and the end of mega-cycles is checked at the beginning  of  Step~$2$.
During each mega-cycle, this algorithm is composed of the following three steps:
\begin{itemize}
\item In Step~$1$, every robot obtains the neighbor's information so that robots can recognize their own lights.
\item In Step~$2$, activated robots in it execute algorithm~$\textit{A}$ using the color of light obtained.
\item In Step~$3$, every robot resets check flags to prepare the next simulation of  algorithm~$\textit{A}$.
Repeating the three steps, the end of the current mega-cycle is checked at the beginning of  Step~$2$ and
if the megacycle is finished, each robot $r$ proceeds to Step~$\RA$ to reset.
\item In Step~$\RA$, every robot resets flag $r.\executed$ indicating whether $r$ has executed in this mega-cycle or not to $\false$ and the control returns to Step~$2$.
\end{itemize}

The details of these steps are as follows:


\begin{itemize}
\item {\bf Step 1-Copy Colors  and Activation} ($\forall \rho\neq r(\rho.\step=1)$). 

In the Look phase, $r$ understands to be in Step~$1$ by detecting  $\rho.\step=1$ for any other robots $\rho$. 
Every robot $r$ stores (and displays) in variable  $r.\CoS.\light$ and $r.\CoS.\executed$, respectively,  its successors' sets of colors and also their sets of executed flags indicating whether the successors' robots have been activated in this mega-cycle.
None of the robots move in this step. 
Note that every robot can recognize  its own color of light and its activated flags in the configuration after Step~$\CCA$ is completed.
This step ends when all robots store  their neighbors' sets of colors and sets of executed flags.
To be able to detect this, every robot sets $r.checked$ to $\true$, indicating that it has executed Step~$1$.
Moreover, whenever all successors of $r$ have set their $checked$ flags to $\true$, $r$ also sets $r.\CoS.checked$ to $\true$.
Now every robot can determine the end of Step~$1$ when the two checked flags ($\rho.me.checked$, $\rho.\CoS.checked$) of every robot $\rho$ become $\true$. 
When this condition is satisfied, $r$ changes $r.\step$ to $2$. 

\item {\bf Step 2-Perform Simulation}  ($\forall \rho\neq r(\rho.\step=2)$). 

In the Look phase, $r$ recognises to be in Step~2 by observing $\rho.\step=2$ of any other robots $\rho$.
Robot $r$ can calculate $r.\light$ and $r.\executed$ by using its predecessor's $\CoS.\light$ and $\CoS.\executed$.  
The rest of Step~$2$ consists of two phases: checking the end of mega-cycles, and  execution of the simulation phase.

{\em Checking End of Mega-Cycles:}
After calculating $r.\light$ and $r.\executed$, $r$ will check if the current mega-cycle has ended. 
If all robots have executed algorithm~$\textit{A}$ in this mega-cycle (all robots have $r.\executed$ set to $\true$, then it moves to Step~$\RA$ (resetting all the executed flags) without changing colors of lights except for the executed flags.

{\em Execution of Simulation :} 
If the mega-cycle has not ended yet, the robots activated in this round will perform a step in \textit{A}.
The activated robots (indicated by set $S_2$) at this time, actually perform their cycle according to algorithm $\textit{A}$
if they have not performed algorithm~$\textit{A}$ in this mega-cycle yet ($r.\executed$ is $\false$).  
The robots that executed algorithm $\textit{A}$ update $r.\executed$ and furthermore update $r.\light$ and move according to algorithm $\textit{A}$. 
After the actual execution of algorithm $\textit{A}$, each robot $r$ performing the simulation sets $r.\step=3$ to reset all checking flags (Step~$3$) to prepare 
for the next Step~$1$.

If all robots in $S_2$ have performed algorithm~$\textit{A}$ so far in this mega-cycle, no light is changed, and Step~$2$ is also performed at the next time.
This situation repeats until some robots not having executed algorithm $\textit{A}$ appear.
The appearance of these robots is guaranteed from the fairness of the scheduler.

\item {\bf Step 3-Reset Check flags}   ($\forall \rho\neq r(\rho.\step =3)$). 

In the Look phase, $r$ understands to be in Step~$3$ by observing $\rho.\step=3$ for any other robots $\rho$. 
In Step~$3$, all robots reset $r.checked$, $r.\CoS.checked$ to $\false$, and $r.\CoS.\executed$ to $\{ \false \}$, enabling the robots to perform another round in this cycle. Each robot resetting its flags changes $r.\step$ to $1$.

\item {\bf Step $\RA$-Reset Mega-Cycle}   ($\forall \rho\neq r(\rho.\step =\RA)$). 
In the Look phase, $r$ understands to be in Step~$\RA$ by observing   $\rho.\step=\RA$ of any other robots $\rho$. 
In Step~$\RA$, 
each robot $r$ sets $r.\executed=\false$ and  $r.\CoS.\executed = \{\false\}$ and the step returns to $2$ to begin a new mega-cycle after all robots reset their activated flags.

\end{itemize}





\Newcodeline
\begin{algorithm}[tp]
\caption{{\tt Scheme-for-modification-flags} for robot $r$ at location $x$.}
\label{algo:mod-scheme}
{\small
\begin{tabbing}
111 \= 11 \= 11 \= 11 \= 11 \= 11 \= 11 \= \kill
{\em Assumptions}\crm
\> Let configuration be $same(\step=\alpha)$ or $except1(\step=\alpha;\gamma)$ 
\crm\crm

{\em State Compute}\crm
\Cl  {\bf if} $\forall \rho \neq r \, (\rho.\step = \alpha$) {\bf then} \` // step $\alpha$\crm
\Cl \> {\bf if} {\bf not}  $\forall \rho \neq r \, (\rho.\flag = \true$)   {\bf then}\crm
\Cl \>\> $r.\step \leftarrow \alpha$\crm
\Cl \>\> $r.\flag \leftarrow \true$\crm

\Cl \> {\bf else if} $\forall \rho \neq r \, (\rho.\flag =\true)$  {\bf then}\crm
\Cl \>\> $r.\flag \leftarrow \true$\crm
\Cl \>\> $r.\step \leftarrow \beta$\crm

\Cl  {\bf else if} $\forall \rho \neq r \, (\rho.\step = \beta$) {\bf then} $r.\step \leftarrow \beta$ \` // step $\beta$\crm
\Cl  {\bf else if} $\forall \rho \neq r \, ((\rho.\step=\alpha)$ {\bf or} $(\rho.\step = \beta))$  {\bf then} $r.\step \leftarrow \beta$
\end{tabbing}
}
\end{algorithm}

Now we show that this simulation algorithm works correctly for $\FC$ robots in \RSY. 
Note that since \RSY\ may start with all robots activated, it is easily validated that this algorithm works well in \FSY(see Appendix for the details).

Since we consider $\FC$ robots, when a robot $r$ checks a predicate, for example $\forall \rho(\rho.\step =\alpha)$,  $r$ cannot see its own $r.\step$. Therefor, $r$ only checks $\forall \rho\neq r(\rho.\step =\alpha)$ and observes that the predicate may be satisfied although $r.\step$ is not $\alpha$. 
This means that predicates appearing in the $\FC$ algorithm must be of the form 
$\forall$ $\rho\neq r$ ($\ldots$)  and we must consider configurations on which only one robot $r$ observes that some predicate holds but any of the other robots observe that the predicate does not hold.

In order to understand our algorithm, consider a simple example where every robot $\rho$ has two lights $\rho.\step (\in \{\alpha, \beta, \gamma\})$ and $\rho.flag (\in \{\true, \false\})$. Let us consider a configuration where $\rho.\step=\alpha$ for any robot $\rho$ and a configuration where there exists just one robot $r_e$ with  $r_e.\step=\gamma$ while $\rho.\step=\alpha$ for any other robot $\rho$.  The former configuration is denoted by  $same(step=\alpha)$ (because all robots are the same step)   and the latter is denoted by $except1(step=\alpha;\gamma)$ (because all robots are  step $\alpha$ except one  which is step $\gamma$).

The next Lemma shows the following: if a step configuration is $same(step=\alpha)$ or $except1(step=\alpha;\gamma)$ with $\rho.flag=\false$ for any robot $\rho$,
Algorithm~\ref{algo:mod-scheme} transforms the step configuration into one of type $same(step=\beta)$ or $except1(step=\beta;\alpha)$ with
$\rho.flag$ changed to $\true$ for all robots $\rho$.
The idea is that as long as a robot still observes other robots with $flag = \false$, robots just keep on setting their $flag$ to $\true$, while staying with $step = \alpha$.
As soon as a robot notices all other robots having $flag = \true$, the robots need to progress to $step = \beta$.
However, the activated robot may not have been active before and therefore still needs to set $flag$ to $\true$.
Now any robot active that sees another robot with $step = \beta$ knows that all $flags$ are $\true$, including its own.

\begin{lemma}\label{lem:scheme-modification}
Let $C(t_a)$ be a step configuration either of type  $same(step=\alpha)$ or $except1(step=\alpha,\gamma \neq \beta)$ at time $t_a$, and let $\forall \rho (\rho$.flag =$\false$) hold at $t_a$.
If Algorithm~\ref{algo:mod-scheme} is executed on $C(t_a)$,  then there is a time $t_b> t_a$ when $C(t_b)$ satisfies the following conditions:
\begin{enumerate}
\item[(a)] The step configuration $C(t_b)$ is either $same(step=\beta)$ or $except1(step=\beta;\alpha)$.
\item[(b)] $\forall \rho (\rho.flag =\true)$ holds at $t_b$. 
\end{enumerate} 
\end{lemma}
\begin{proof}
Consider the two cases:
(1) $same(step=\alpha)$ and 
(2) $except1(step=\alpha,\gamma \neq \beta)$ at $t_a$.
 
(1) $same$(step=$\alpha$) at $t_a$: 
Since line~1 holds and line~2 holds until $\forall \rho (\rho.flag =\true)$ holds, the number of robots $\rho$ with $\rho.flag=\false$ (denoted as $\#(flag=\false))$ decreases. 
As long as $\#(flag=\false)\geq2$,  since line~1 and line~2 hold and the schedule is fair,  $\#(flag=\false)$ decreases. 
Let $t$ be the first time when $\#(flag=\false)\leq 1$. 

\begin{itemize}
\item 
If $\#(flag=\false)=0$, since $\forall \rho (\rho.flag =\true)$ holds,
the  robots activated at $t$ execute lines~6 and~7 and the step configuration becomes $\{\alpha, \beta\}$ at $t+1$. 
If $\#(step=\alpha)\geq 2$, the
 robots activated after $t+1$ execute line~9 and  $\#(step=\alpha)$ decreases.
As long as $\#(step=\alpha)\geq 2$, when robots $r$ with $r.\step=\alpha$ are activated, they observe $\rho.\step=\alpha$ and $\rho'.\step=\beta$.
Then, since there is a time when $\#(step=\alpha)\leq 1$, letting $t_b$ be such time,  the lemma holds.

\item
Otherwise ($\#(flag=\false)=1$),  let $r_e$ be the robot with $r_e.flag=\false$.
In this case, since $r_e$ observes $\forall \rho\neq r_e(\rho.flag = \true)$, it will execute lines~5-7, changing $r_e.\step$ to $\beta$, but also setting $r_e.flag$ to $\true$.
The next time, $\forall \rho (\rho.flag =\true)$ holds. Therefore, the same actions as in the previous situation of $\#(flag=\false)=0$ are taken, and the lemma holds.
\end{itemize}

 (2) $except1(step=\alpha,\gamma \neq \beta)$ at $t_a$:
 Let $r_e$ be the robot with $r_e.\step=\gamma$ and let $t$ be the first time $r_e$ is activated after $t_a$.
Any activated robot except r$_e$ between $t_a$ and $t$ does not change the color of its own lights, because it observes $r_e.\step=\gamma$.
Since only $r_e$ observes $\forall \rho \neq r_e(\rho.\step =\alpha$) and {\bf not}  $\forall \rho \neq r(\rho.flag =\true)$,
$r_e$ changes $r_e.\step=\alpha$ and the step configuration becomes $same(step=\alpha)$  at $t+1$.
Then, this case can be reduced to case (1) and the lemma holds.
\end{proof}

This scheme of transitions between step configurations is used in the algorithm.
In particular, transitions from Step~$1$  to Step~ $2$, from Step~$2$ to Step~$m$, from Step~$m$ to Step~$2$, and from Step~$3$ to Step~$1$ (Lemmas~\ref{lem:step-1-2}, \ref{lem:step-2-4}, \ref{lem:step-4-2} and~\ref{lem:step-3-1}).  Note that the condition of \RSY\ is not used in the proof of Lemma~\ref{lem:scheme-modification}, that is, Algorithm~\ref{algo:mod-scheme} works also in  \SSY\  correctly. The prohibited pattern in \RSY\ is necessary only for the actual simulation (Step~2).

Let $\pi_1(X,t)$ denote the predicate {\em checked-all-neighbors-and-me}: \\
$\forall$ $\rho \in X$($\rho.me.checked = \true$ {\bf and} $\rho.\CoS.checked=\true$), at time $t$.

\begin{lemma}\label{lem:step-1-2}
Assume that step configuration is $same(step=1)$ or $except1(step=1;3)$ and it holds that $\forall \rho \in R(\rho.checked =\false$ and $\rho.\CoS.checked=\false)$ at time $t_1$.
Then there exists a time $t_2>t_1$ such that the following (1) and (2) hold.
\begin{enumerate}
\item[(1)] The step configuration at $t_2$ is $same(step=2)$ or $except1(step=2;1)$.
\item[(2)] $\pi_1(R,t_2)$ holds and  for every robot $r \in R$ at any location $x$,
 $r.\CoS.\light = \{ \rho.\lcolor \mid \rho \in \CoS(x)\}$ and
 $r.\CoS.\executed =\{ \rho.\executed \mid \rho \in \CoS(x)\}$ at $t_2$.
\end{enumerate}
\end{lemma}

Lemma~\ref{lem:step-1-2} implies that each robot $r$ can calculate its own color ($r.\light$) and whether $r$ has executed algorithm~$\textit{A}$ in the current mega-cycle or not ($r.\executed$) by looking at the neighbor's corresponding lights.


The color of robot $r$ ($r.\lcolor$) used in the execution of algorithm~{\textit A} is determined as follows:
Assuming all robots have correctly set $r.\CoS.\light$, $\rho.\CoS.\light$ for any $\rho\in \CoP(x)$ contains all colors at location $x$. Let $r.\light.\here$ be the set of colors seen at $x$. Note that this by definition does not include the color of $r.\light$, since $r$ is on location $x$ and it cannot see its own light. Now $r$ can calculate the color of $r.\light$ in the following way:
$$ r.\light = \CoP(x).\CoS.\light  - x.\light.\here.$$

Since the executed flag of robot $r$ contains 2 colors, $r.\executed$ can be determined similarly to the case of determining $r.\light$ by using $\CoP(x).\CoS.\executed$ instead of $\CoP(x).\CoS.\light$.
We can remove the assumption of chirality by using the method of \cite{FSW19}.
Although without chirality a unique circular ordering of the locations $X(t)$ occupied by the robots cannot be obtained, we can still define a set of neighboring locations.
Instead of having a single flag $r.\CoS.\light$ indicating a single set of colors seen at $\CoS(x)$, every robot has a flag $r.neigh.\light$, which is a pair of sets. One of those sets contains the colors at $\CoS(x)$, and the other one the colors at $\CoP(x)$.
In this way, every robot can still determine its own color.


The following Lemmas~\ref{lem:step-2-4}, \ref{lem:step-4-2} and~\ref{lem:step-3-1} show that transitions from Step~$2$ to Step~$m$, from Step~$\RA$ to Step~$2$, and from Step~$3$ to Step~$1$ work correctly, respectively due to Lemma~\ref{lem:scheme-modification}.

\begin{lemma}\label{lem:step-2-4}
Assume that step configuration is $same(step=2)$ or $except1(step=2;1)$ and it holds that $\forall \rho \in R(\rho.\executed=\true)$ at time $t_2$.
Then there exists a time $t_\RA>t_2$ such that the following (1) and (2) hold.
\begin{enumerate}
\item[(1)] The step configuration at $t_m$ is $same(step=\RA)$ or $except1(step=\RA;2)$.
\item[(2)] $\forall$ $\rho \in R$($\rho.\executed =\true$) at time $t_\RA$
\end{enumerate}
\end{lemma}

\begin{lemma}\label{lem:step-4-2}
Assume that step configuration is $same(step=\RA)$ or $except1(step=\RA;2)$ and it holds that $\forall \rho \in R(\rho.\executed =\true)$ at time $t_\RA$
Then there exists a time $t_2>t_\RA$ such that the following (1) and (2) hold.
\begin{enumerate}
\item[(1)] The step configuration at $t_2$ is $same(step=2)$ or $except1(step=2;\RA)$.
\item[(2)] $\forall$ $\rho \in R$($\rho.\executed =\false$) and $(\rho.\CoS.\executed =\{\false\})$ at time $t_2$

\end{enumerate}
\end{lemma}


\begin{lemma}\label{lem:step-3-1}
Assume that step configuration is $same(step=3)$ or $except1(step=3;2)$
Then there exists a time $t_1>t_3$ such that the following (1)-(3) hold.
\begin{enumerate}
\item[(1)] The step configuration at $t_1$ is $same(step=1)$ or $except1(step=1;3)$.
\item[(2)] $\forall \rho \in R (\rho.\lcolor$ and $\rho.activated$ at $t_1$ are the same as these at $t_3$).
\item[(3)] $\forall$ $\rho \in R$( the lights except $\rho.$light and $\rho.$executed are reset to $\false$ at $t_1$).
\end{enumerate}
\end{lemma}


Therefore, remaining is the correctness of the transition from Step~$2$ to Step~$3$, including the execution of algorithm \textit{A}.
In Step~$2$, if this mega-cycle is not finished, the step configuration is either ($same(step=2)$ or  $except1(step=2;1)$),  or ($same(step=2)$ or $except1(step=2;\RA)$). We show that the actual execution of the simulated algorithm~${\textit A}$ is performed correctly in Step~$2$ if  \RSY\ is assumed in the following lemma.

\begin{lemma}\label{lem:step-2}
Assume that step configuration is $same(step=2)$ or $except1(step=2;1))$ or ($same(step=2)$ or $except1(step=2;\RA))$ at time $t_2$.
If it  does not hold that $\forall \rho \in R(\rho.\executed =\true)$, 
then there exists a time $t_3>t_2$ such that the following (1)-(3) hold. 
Let $R_{\true}(t)$ be $\{r | r.\executed = \true \text{ at time } t\}$.
\begin{enumerate}
\item[(1)] The step configuration at $t_3$ is $same(step=3)$ or $except1(step=3;2)$.
\item[(2)] There exists just one time $t_a \in [t_2,t_3)$ at which all robots $r$ that are element of the non-empty set $S \subseteq R\setminus R_{\true}(t_2)$ execute algorithm~{\textit A}.
\item[(3)] $R_{\true}(t_3) = R_{\true}(t_2)\cup S$.  
\end{enumerate}
\end{lemma}

\begin{proof}
Let $S_2$ be a set of robots activated at $t_2$. Note  that $n >|R_{\true}(t_2)| \geq 0$.
There are two cases  for the step configuration:
(Case 1) $same(\step=2)$,  and  (Case 2)  ($except1(\step=2;1))$ or $except1(\step=2;\RA))$.

\noindent
{\bf Case 1:}  The step configuration is $same(\step=2)$.

If $S_2=R$ (called full activation), full activation must have always occurred   from the starting time to $t_2$ because of the definition of   \RSY.
Then, since it holds that $|R_{\true}(t_2)| = 0$ and any robot $r$ in $S_2$ observes that $r.\executed=\false$ (Lemma~\ref{lem:step-4-2}),
$r$ executes algorithm~$A$ and sets $r.\executed$ and $r.\step$ to $\true$ and $3$, respectively.  Thus, all robots activated at the next time $t_2+1$ observe some robot $\rho$ with $\rho.\step=3$, and there exists a time $t_3$ such that the step configuration at $t_3$ is $all(\step=3)$ or $except1(\step=3;2)$ by a similar way as in the proof of Lemma~\ref{lem:scheme-modification}. Setting $t_a=t_2$,
the conditions of the lemma hold.

Otherwise ($S_2\subset R$). Some robots in $S_2$ have their $execute$ flag set to $\true$, while some have it set to $\false$. Let $\tilde{S_2}$ be a set of robots that executed algorithm~$A$ at time $t_2$, i.e. the subset of robots from $S_2$ that had their $execute$ flag set to $\false$.
There are two cases depending on whether $\tilde{S_2}$ is empty or not.
If $\tilde{S_2} \neq \emptyset$, since the set of robots activated at $t_2+1$ (denoted as $S'_2$) and $S_2$ are disjoint due to \RSY, 
all robots in $S'_2$ observe some robot $\rho$ with $\rho.\step=3$. Now, due to Lemma~\ref{lem:scheme-modification}, there exists a time $t_3$ such that the step configuration at $t_3$ is $all(\step=3)$ or $except1(\step=3;2)$. Setting $t_a=t_2$, the conditions of the lemma are satisfied.

If $\tilde{S_2} =\emptyset$, since the configuration at $t_2$ will not be changed after $t_2$ and the scheduler is fair, there exists a time $t >t_2$ such that $\tilde(S_t) \neq \emptyset$, where $S_t$ is a set of robots activated at $t$. Then the lemma holds similarly to the former case.

{\bf Case 2:}  The step configuration is $except1(\step=2;1))$ or\\ $except1(\step=2;\RA))$.

Let $r$ be a robot such that $r.\step=1$ or $r.\step=\RA$.
None of the robots (except $r$) can execute Step~$2$ because they observe $r$ with $r.\step \neq 2$.
As long as $r$ is not activated, the configuration is  unchanged.
Let $t$ be the first time when $r$ is activated after $t_2$ and let $S_t$ be a set of robots activated at $t$.
If $r.\executed=\true$ at $t$, $r,\step$ becomes $2$ and the step configuration becomes $same(\step=2)$.
Then, this case is reduced to {\bf Case 1}.
Otherwise ($r.\executed=\false$), $r$ executes algorithm~$A$ and sets $r.activated$ and $r.\step$ to $\false$ and $3$.
Note that any robot in $S_t \setminus \{r\}$ cannot execute Step~$2$ at $t$. Setting $t_a=t$, the Lemma~\ref{lem:scheme-modification} holds.

\end{proof}

By  Lemmas~\ref{lem:step-1-2}-\ref{lem:step-2}, \textit{sim-LUMI-by-FCOM(A)} executes Step~$1$-Step~$3$ and Step~$\RA$ in infinite rounds in \RSY\ and the execution of \textit{A} obeys \SSY.
Let $E$ be the sequence of the set of activated robots that execute simulated algorithm~{\textit A} in Algorithm \textit{sim-LUMI-by-FCOM(A)}. Since, by Lemma~\ref{lem:step-2},  any mega-cycle is completed and every robot executes exactly once in every mega-cycle, $E$ is fair.
Then we have obtained Theorem~\ref{th:sim-LUMI-by-FCOM}. Note that, if   algorithm~\textit{A} uses $\ell$ colors, the simulating algorithm~ \textit{sim-LUMI-by-FCOM(A)} uses $O(\ell 2^{\ell})$ colors.

Therefore it holds that $\mathcal{FCOM}^{RS} \geq \mathcal{LUMI}^{S}$.
Since $\mathcal{LUMI}^{S}\equiv \LU^{RS}$ by Theorem~\ref{th:LUMISeqLUMIRS}, we have that: 
$\mathcal{FCOM}^{RS} \geq \mathcal{LUMI}^{RS}$. Moreover, 
since the reverse relation is trivial and $\FS^{RS} <\LU^{S}$ (Theorem~\ref{th:RSvsS2}(1)), 
the next theorem follows. 

\begin{theorem}\label{th:FCRSLURSLUS}
 $\FC^{RS} \equiv \LU^{RS} \equiv \LU^{S}$, and \\
  $\FC^{RS} > \FS^{RS} > \OB^{RS}$.
 \end{theorem}
\begin{proof}
The equality follows from Theorems~\ref{th:LUMISeqLUMIRS} and~\ref{th:sim-LUMI-by-FCOM}.\\
$\FC^{RS} > \FS^{RS}$ follows from that 
$\LU^{S} \equiv \FC^{RS}$ and $\LU^{S} > \FS^{RS} $ (Theorem~\ref{th:RSvsS2}(1)).
$\FS^{RS} > \OB^{RS}$ follows from RDV can be solved in $\FS^{RS}$ and
Lemma~\ref{lem:RDV}.
\end{proof}

%



 
\section{Concluding Remarks}
\label{sec:conclusion}
In this paper, we have  started the
 investigation of  the algorithmic and computational issues arising in distributed systems 
 of autonomous mobile entities in the Euclidean plane  where their energy is limited, albeit renewable.
 
We have   studied   the difference  in computational power caused by the energy restriction, and provided
a complete characterization of the computational difference between  energy-restricted 
and unrestricted robots
in all four models considered in the literature: $\OB, \FS, \FC,\\ \LU$
(see Fig. \ref{fig:RelationshipBetweenRandS}).
We have also examined the difference with energy-unrestricted robots operating under a fully-synchronous
scheduler (see Fig. \ref{fig:RelationshipBetweenFandR}). Furthermore,  we have studied   the impact that memory persistence 
and communication capabilities have on the computational power
of such energy-constrained systems of robots.
Some of these results have been obtained through the design and analysis of novel {\em simulators}:
 algorithms that allow a set of robots with a given set of capabilities to execute correctly any protocol
 designed for robots with more powerful capabilities.

The current contribution  provides a complete answer to the problems considered. At the same time,
it opens an entirely new research area.  
Indeed, many other forms of energy constrains can be considered, and their computational
impact analyzed.
In particular, the case considered here, of a robot consuming all its energy after $p=1$  activity cycles and requiring
$q=1$  idle rounds for energy restoration, opens the investigation 
for $(p,q)$-energy constrained  robots.

\begin{figure}[H]
\centering\includegraphics[keepaspectratio, width=13cm]{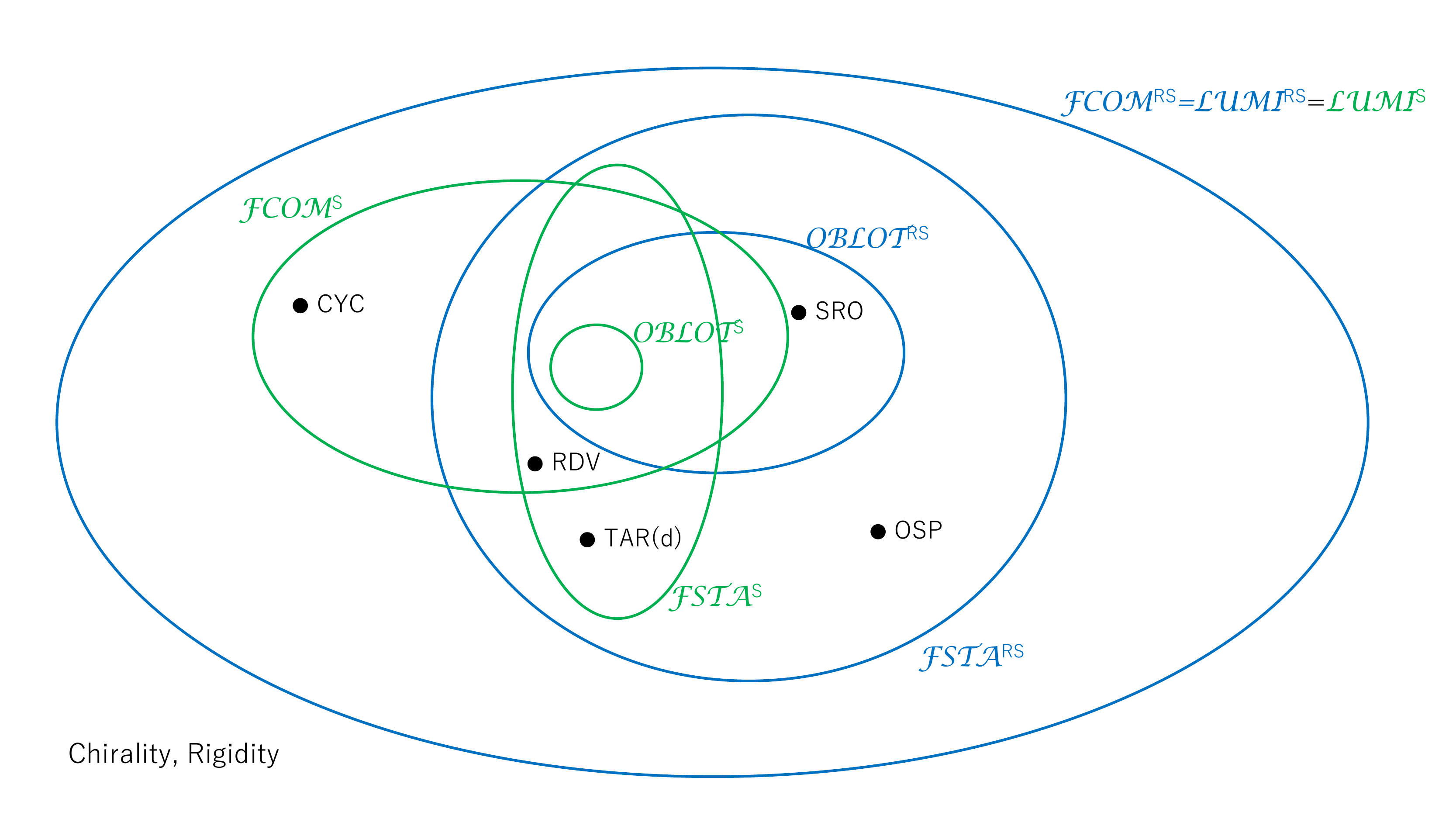}
    \caption{Relationship among $\LU$, $\FC$, $\FS$ and $\OB$ in \RSY\ and  \SSY\ assuming   chirality  and  rigidity, where TAR(d) is defined in~\cite{FSW19}.}
    \label{fig:RelationshipBetweenRandS}

\end{figure}

\begin{figure}[H]
\centering\includegraphics[keepaspectratio, width=13cm]{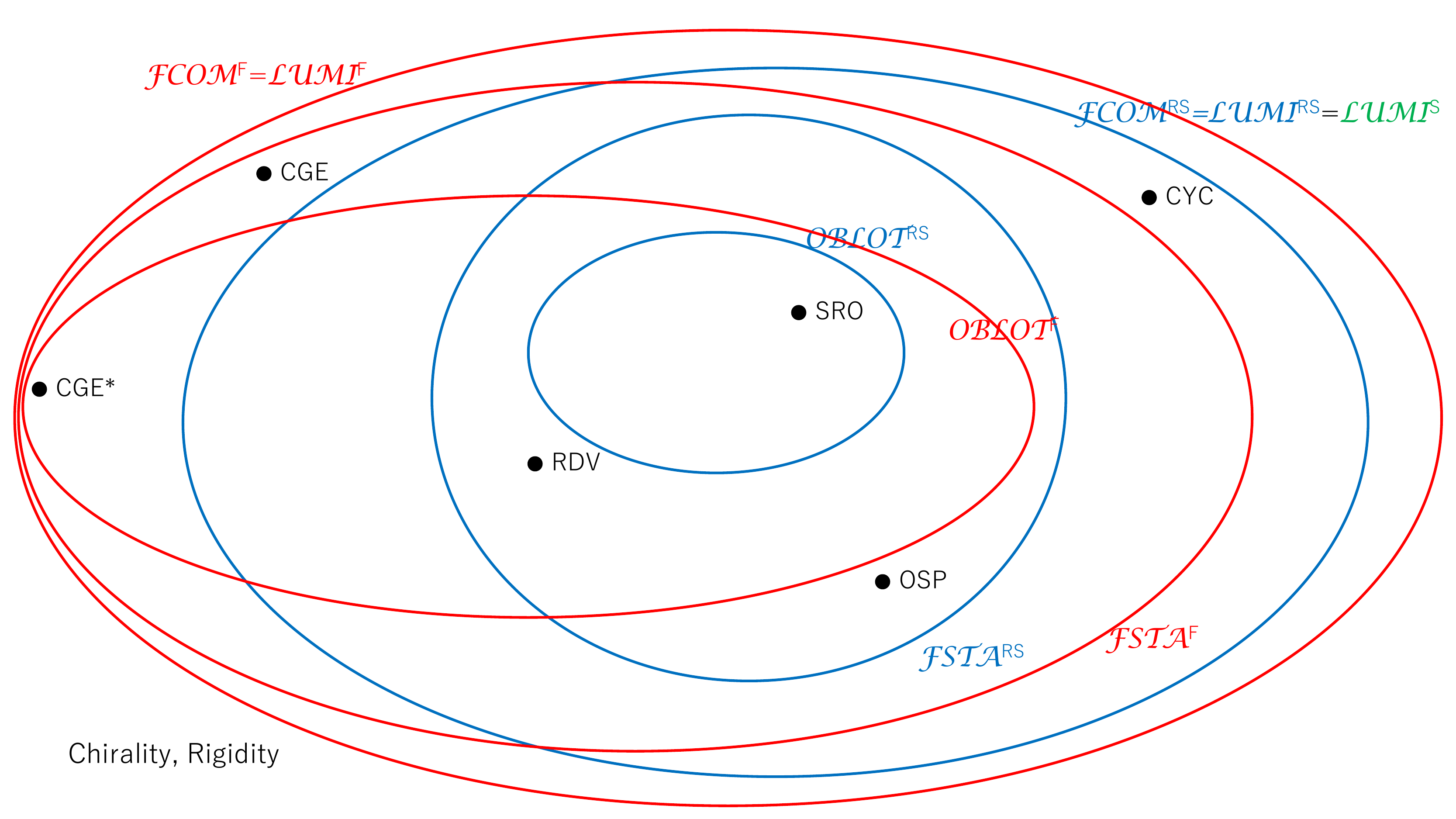}
    \caption{Relationship among $\LU$, $\FC$, $\FS$ and $\OB$ in \FSY\ and  \RSY\ assuming   chirality  and  rigidity.}
    \label{fig:RelationshipBetweenFandR}

\end{figure}

\newpage

\bibliographystyle{plain}
\bibliography{referencesorg}

\end{document}